\documentclass{article}  
\usepackage[margin=1in,headheight=25pt]{geometry}
\usepackage{booktabs}
\usepackage{amsmath}
\usepackage{amssymb}
\usepackage{times}  
\usepackage{helvet}  
\usepackage{courier}  
\usepackage{url}  
\usepackage{graphicx}  
\usepackage{bm}
\usepackage{amsthm}
\usepackage{breqn}
\usepackage{mathtools}
\usepackage{algorithm}
\usepackage{algorithmicx}
\usepackage[noend]{algpseudocode}

\newcommand*\Let[2]{\State #1 $\gets$ #2}

\DeclareMathOperator*{\E}{\mathbb{E}}
\DeclareMathOperator*{\pr}{\text{Pr}}
\newcommand*{\pov}{\text{POV}}
\newcommand*{\mov}{\text{MOV}}

\newtheorem{theorem}{Theorem}

\theoremstyle{definition}

\begin{document}

\title{Controlling Elections through Social Influence}  



\author{Bryan Wilder\\Department of Computer Science and Center for Artificial Intelligence in Society\\University of Southern California\\bwilder@usc.edu \and  Yevgeniy Vorobeychik\\Department of Electrical Engineering and Computer Science\\Vanderbilt University\\yevgeniy.vorobeychik@vanderbilt.edu}  
\date{}
\maketitle
\begin{abstract}  
Election control considers the problem of an adversary who attempts to tamper with a voting process, in order to either ensure that their favored candidate wins (constructive control) or another candidate loses (destructive control). 
As online social networks have become significant sources of information for potential voters, a new tool in an attacker's arsenal is to effect control by harnessing social influence, for example, by spreading fake news and other forms of misinformation through online social media.

We consider the computational problem of election control via social influence, studying the conditions under which finding good adversarial strategies is computationally feasible. We consider two objectives for the adversary in both the constructive and destructive control settings: probability and margin of victory (POV and MOV, respectively). 
We present several strong negative results, showing, for example, that the problem of maximizing POV is inapproximable for any constant factor.
On the other hand, we present approximation algorithms which provide somewhat weaker approximation guarantees, such as bicriteria approximations for the POV objective and constant-factor approximations for MOV.
Finally, we present mixed integer programming formulations for these problems.
Experimental results show that our approximation algorithms often find near-optimal control strategies, indicating that election control through social influence is a salient threat to election integrity.
\end{abstract}

%

\section{Introduction}

The integrity of elections is crucial to the functioning of democratic
institutions. As a result, a large body of work has focused on the
robustness of elections to various forms of control, where a malicious
party attempts to manipulate election results, for example, by bribing
voters and adding or removing votes. While it is important to
understand the vulnerability of elections to such control, there are
many countries where blatant tampering is (fortunately) uncommon. For
instance, outright voter fraud is very rare in the United States federal
and state elections
\cite{ahlquist2014alien,christensen2014identifying}. 

However, more subtle forms of election control may attempt to subvert
legitimate information channels towards malicious means.
For example, political advertising and news (in the editorial form) are common
legitimate means for convincing prospective voters. 
Such communication, when sufficiently transparent, is often critical to the effective functioning of
democracy, and can exert considerable influence on voting behavior \cite{gerber2009does,dellavigna2007fox,brader2005striking}.
Malicious control over information promulgated through these channels
can thus have considerable impact, but is also difficult to achieve
due to the relative transparency of traditional media sources.

The increasing importance of social media, such as Facebook and Twitter, for propagating information,
including about political issues
\cite{chi2011twitter,wattal2010web,holcomb2013news}, is a game
changer.
Both the decentralized nature of information sources on social media,
and their associated lack of transparency, present malicious parties 
with an unprecedented opportunity to influence a democratic
political process.
Recent evidence of deliberate election tampering in the 2016 US
presidential election through \emph{fake news}---deliberately
falsified news content---suggests that this
issue is a major concern for election integerity for years
to come~\cite{allcott2017social}.
For example, it is estimated that the typical American adult saw at least one fake news story during the election cycle \cite{allcott2017social}, and such stories have been shown to impact voters' judgment \cite{pennycook2017prior}.


Motivated by these concerns, we initiate the first algorithmic study
of the problem of \emph{election control through social influence}.
In our setting, there is a social network of voters who elect a single
winner by plurality vote. An outside party may select a subset of
nodes as \emph{seed nodes} for a news story or advertisement. Each of
these seed nodes shares the story with their friends. Each friend has
some probability of being influenced in their voting preferences, as
well as sharing the story further. The question is whether, given a limited budget, the attacker can influence enough voters to ensure that a target candidate wins or loses the election. 

This problem is closely related to \emph{influence maximization},
which has been studied primarily in the context of viral
advertising. There, the objective is simply to maximize the expected
number of people who receive a message. While influence
maximization admits a simple $(1 - 1/e)$-approximation algorithm,
election control through social influence presents a number of new
algorithmic challenges. We study both constructive and destructive
control for two different objectives: 1) maximizing the expected
margin by which a target candidate wins (loses) the election
(\emph{margin of victory}, or MOV), and 2) maximizing the probability
that a target candidate wins (loses) election (\emph{probability of
  victory}, or POV). 


\noindent{\bf Summary of main results:} 
We provide a mix of negative (hardness and inapproximability) and
positive (algorithmic) results for the problem of election control
through social influence.
Our main contributions are the following:

\begin{itemize}
	\item We show that the MOV objective in the two candidate case is monotone submodular and hence admits a $(1 - 1/e)$ greedy approximation algorithm. 
	\item We prove that the POV objective is hard to approximate
          to within \emph{any} multiplicative factor for both
          constructive and destructive control, even in elections with
          only two candidates. 
	\item We provide a bicriteria approximation algorithm for the
          POV objective in the two-candidate case.
In fact, our algorithm applies to the more general problem of maximizing the probability that a submodular function exceeds a given threshold value and may be of general interest. 
	\item In the multicandidate case, we provide algorithms which
          achieve similar guarantees as the two-candidate case up to
          the loss of a constant factor (independent of
          the number of candidates). Such guarantees hold for both
          constructive and destructive control, for both the MOV and POV objectives. 
	\item We give mixed integer linear programming (MILP) formulations for all of the above settings which can be used to find optimal solutions. 
	\item We experimentally compare our approximation algorithms
          to the optimal strategies produced by the MILP. Despite
          formal hardness results, our approximation algorithms often
          find near-optimal solutions, particularly for the MOV objective. This suggests that computational hardness may not always be a practical barrier to controlling elections via social influence. 
\end{itemize}

\section{Related work}

Our work is closely related to two research areas:
election control and influence maximization. These bodies of work are
separate: to our knowledge there is no previous work which considers
election control \emph{using} social influence. The computational
study of election control was started by Bartholdi et al.\
\cite{bartholdi1992hard}, who considered constructive control. The
destructive control setting was introduced by Hemaspaandra et al.\
\cite{hemaspaandra2007anyone}. A large body of work has studied
election control under different settings and voting rules
\cite{conitzer2007elections,erdelyi2009sincere} including bribing
voters
\cite{faliszewski2009hard,faliszewski2009llull,baumeister2015complexity,erdelyi2017complexity,yang2016hard},
adding or deleting voters
\cite{erdelyi2015more,loreggia2015controlling,faliszewski2011multimode,liu2009parameterized},
and adding or deleting candidates
\cite{chen2015elections,faliszewski2011multimode,liu2009parameterized}. Another
topic is strategic behavior on the part of the voters themselves \cite{meir2008complexity,obraztsova2012optimal,obraztsova2013manipulation}. The main difference between our work and previous work on election control is that we introduce and analyze social influence as a novel mechanism for both constructive and destructive control.

There is a large, parallel literature devoted to influence maximization in social networks. This line of work was introduced by Kempe et al. \cite{kempe2003maximizing} who introduced influence maximization in the independent cascade model and proposed a greedy algorithm based on submodularity. Since then, a number of newer algorithms have been proposed, mostly attempting to scale up the greedy algorithm to very large graphs \cite{chen2010scalable,cohen2014sketch,tang2014influence,goyal2011celf++}. Subsequent work has also introduced several related settings, e.g., continuous time dynamics \cite{rodriguez2012influence,du2013scalable}, bandit settings where dynamics are learned over time \cite{chen2013combinatorial,lei2015online,carpentier2016revealing}, or robust problems where influence probabilities are uncertain \cite{heRobust2016,chen2016robust,lowalekar2016robust,wilder2017uncharted}. None of this work considers using social influence to control an election, and our setting brings a range of new technical challenges. Almost all work on influence maximization is founded on submodularity of the objective function. However, even though we use the same model of influence spread, objectives related to election control often violate submodularity, and we need to develop new algorithmic techniques. We mention here work by Krause et al.\ \cite{krause2008robust} on robust submodular optimization. For optimizing the POV objective, we use a similar form of surrogate objective. However, our objective is to maximize the \emph{probability} of a desired outcome, not the worst-case value, so both our final algorithm and analysis are novel.

\section{Problem formulation}

We consider an election with candidates $C = \{c_*, c_1, c_2, ... c_\ell\}$. $c_*$ is a special target candidate, and the objective of the election control problem is to make $c_*$ either win the election (constructive control) or lose (destructive control). The voters are the nodes of a graph $G = (V, E)$. Each voter $v$ has an ordering $\pi_v$ over the candidates and casts a vote for $\pi_v(1)$, i.e., their first ranked candidate. We assume that voters do not behave strategically. The winner is decided via the plurality rule (the candidate with the most votes wins the election). If there is a tie, we say that the attacker fails. This tie-breaking does not impact any of our results. Let $V_{c_i}^{j} = \{v \in V: \pi_v(j) = c_i\}$ be the set of voters who rank candidate $c_i$ in place $j$. Initially, $c_i$ has $|V_{c_i}^1|$ votes. 

\textbf{Social influence: }There is an attacker who wishes to change the results of the election by spreading messages which cause voters to change their ordering over candidates. In \emph{constructive control}, the attacker can spread a message which causes any voter $v$ who becomes influenced to promote $c_*$ by one place in $\pi_v$ (exchanging $c_*$ with the candidate previously ranked above them). If $\pi_v(c_*) = 1$, the message has no effect on $v$, but $v$ may still decide to share the message with their neighbors. In \emph{destructive control}, a voter who is influenced demotes $c_*$ by one place in $\pi_v$. Influence spreads via the independent cascade model (ICM), the most common model in the influence maximization literature. Each edge $(u,v) \in E$ has a propagation probability $p_{u,v}$. If $u$ is influenced, it makes one attempt to influence each neighbor $v$. Each attempt succeeds independently with probability $p_{u,v}$. The attacker may select a set of $k$ \emph{seed nodes} who are influenced at the start of the process. The diffusion then proceeds in discrete time steps until no new activations are made. 

We also introduce a useful alternate view of the ICM, the \emph{live-graph} model. We can equivalently see the ICM as removing each edge $(u,v)$ from the graph with probability $1 - p_{u,v}$. A node is influenced if it is reachable from any seed node via the edges that remain. Call any specific setting of present/absent edges a scenario $y$, with induced graph $G_y$. Let $m = 2^{|E|}$ be the total number of scenarios. Let $f(S, y)$ denote the number of nodes which are reachable from any seed node in $S$ on graph $G_y$. The expected number of nodes influenced under the ICM is just $f(S) = \E_y\left[f(S, y)\right]$. Similarly, the probability that the number of influenced nodes exceeds any threshold value $\Delta$ is $\pr_y\left[f(S, y) \geq \Delta\right]$. At times, we will want to specifically reason about the probability some subset of $V$ is influenced. For any $A \subseteq V$, let $f(S, y, A)$ denote the number of nodes in $A$ reachable from $S$ in scenario $y$. Analogously, $f(S, A) = \E_y[f(S, y, A)]$. We remark that such functions can be evaluated up to arbitrary precision by averaging over random samples for $y$. For simplicity, we ignore such issues here since they are well understood \cite{tang2014influence,chen2010scalable,cohen2014sketch}.

\textbf{Objectives: }We now formally introduce our two objectives, starting with the simpler two-candidate case. In a two-candidate election,  constructive and destructive control are clearly equivalent since maximizing the probability that $c_*$ loses is the same as maximizing the probability that the other candidate wins (and vice versa). Hence, we study only constructive control without loss of generality.

In the margin of victory (MOV) objective, we want to maximize the expected number of votes by which $c_*$ wins the election. We define our objective as the \emph{change} in the expected margin:

\begin{align*}
\mov(S) = 2 \E_y\left[f(S, y, V_{c_*}^2)\right].
\end{align*}

The factor 2 is present since reaching a voter in $V_{c_*}^2$ both adds a vote for $c_*$ and removes a vote for the opponent. We study the expected change in the margin (not the margin itself) so that approximation ratios are well defined even when the margin is negative.

In the probability of victory (POV) objective, we want to maximize the probability that $c_*$ wins the election. Let $\Delta = \frac{1}{2}\left(|V_{c_1}^1| - |V_{c_*}^1|\right) +1$ be the number of voters that $c_*$ needs to reach in order to win the election. The POV objective is 

\begin{align*}
\pov(S) = \pr_y\left[f(S, y, V_{c_*}^2) \geq \Delta\right]
\end{align*}

which is just the probability that at least $\Delta$ of the voters who have $c_*$ in second place are reached. 

In the multicandidate case, constructive and destructive control are no longer equivalent. Further, the impact of messages is more complex than before. E.g., in constructive control not only does $c_*$ gain a vote, but another candidates loses a vote; we need to keep track of the number of votes lost by each other candidate.

We start out by defining functions which give the change in the margin between $c_*$ and another candidate $c_i$ in a given scenario $y$ when seed set $S$ is chosen. Let $\chi(v, S, y)$ be 1 if node $v$ is reachable from seed set $S$ in the graph $G_y$. The change in margin (in constructive and destructive control, respectively) is given by 

\begin{align*}
&g_C(S, y, c_i) = \sum_{v \in V_{c_*}^2 \setminus V_{c_i}^1} \chi(v, S, y) + 2\sum_{v \in V_{c_*}^2 \cap V_{c_i}^1} \chi(v, S, y)\\
&g_D(S, y, c_i) = \sum_{v \in V_{c_*}^1 \setminus V_{c_i}^2} \chi(v, S, y) + 2\sum_{v \in V_{c_*}^1 \cap V_{c_i}^2} \chi(v, S, y)
\end{align*}

which gives value 2 for every node that is flipped from $c_i$ to $c_*$ (or vice versa) and hence count double towards the margin, and 1 for other nodes. Based on this, we now give expressions for the change in margin given any fixed scenario $y$ and seed set $S$. We start with constructive control. Note that before any intervention, the margin is just $\max_{c_i} |V_{c_i}^1| - |V_{c_*}^1|$. Afterwards, the margin is $\max_{c_j} |V_{c_j}^1| - g_C(S, y, c_j) - |V_{c_*}^1|$. Hence, the change in margin is 

\begin{align*}
m_C(S, y) &= \left[\max_{c_i} |V_{c_i}^1| - |V_{c_*}^1|\right] - \left[\max_{c_j} |V_{c_j}^1| - g_C(S, y, c_j) - |V_{c_*}^1|\right]\\
&= \min_{c_j} \left(g_C(S, y, c_j) + \max_{c_i} |V_{c_i}^1| - \left|V_{c_j}^1\right|\right).
\end{align*}

That is, the change in margin is driven by candidate with largest starting vote ($|V_{c_j}^1|$) and smallest loss in vote $(g_C(S, y, c_j))$. Now considering all scenarios $y$, the constructive control objectives are

\begin{align*}
\mov_C(S) = \E_y\left[m_C(S, y)\right] \quad \pov_C(S) = \pr_y\left[m_C(S, y) \geq \Delta_C\right].
\end{align*}

where $\Delta_C = \max_{c_i} |V_{c_i}^1| - |V_{c_*}^1| + 1$ is the necessary change in margin for $c_*$ to win. For destructive control, we can similarly write the change in margin and corresponding objectives as 

\begin{align*}
&m_D(S, y) = \max_{c_i} \left(g_D(S, y, c_i) + |V_{c_i}^1|\right) - \max_{c_j} |V_{c_j}^1|\\
&\mov_D(S) = \E_y\left[m_D(S, y)\right] \quad\quad \pov_D(S) = \pr_y\left[m_D(S, y) \geq  \Delta_D \right].
\end{align*}

\noindent where $\Delta_D = |V_{c_*}^1| - \max_{c_i} |V_{c_i}^1| + 1$.

%

%
%
%

\section{Elections with two candidates}

We start with elections with only two candidates. Recall that in this setting, constructive and destructive control are equivalent, so our results are stated only for constructive control (trying to ensure $c_*$ wins the election). In order to state our algorithmic results, we first introduce some background on submodular optimization and influence maximization. A set function $f : V \to R$ is submodular if for all $A \subseteq B \subseteq V$ and all $x \not\in B$, $f(A \cup \{x\}) - f(A) \geq f(B \cup \{x\}) - f(B)$. Intuitively, submodularity formalizes the property of diminishing returns. The function $f(S)$ which gives the expected number of nodes reached by $S$ under the independent cascade model is known to be monotone submodular. It is well known that whenever $f$ is a monotone submodular function, the greedy algorithm gives a $(1 - 1/e)$-approximation to the problem $\max_{|S| \leq k} f(S)$.

It is natural to hope that submodularity would transfer to our election control objectives $\mov$ and $\pov$.  Our first result is that submodularity does in fact hold for the $\mov$ objective with two candidates. Previous results for influence maximization do not directly apply because the $\mov$ objective only counts nodes who have $c_*$ in second place. Nevertheless, similar reasoning applies.
\begin{algorithm}
	\caption{Algorithms for MOV objective} \label{alg:mov}
	\begin{algorithmic}[1] 
		\Function{Greedy}{$h$, $k$}
		\Let{$S$}{$\emptyset$}
		\While{$|S| < K$}
		\Let{$v$}{$\arg\max_{v \in V \setminus S} h(S \cup \{v\}) - h(S)$}
		\Let{$S$}{$S \cup \{v\}$}
		\EndWhile
		\State\Return $S$
		\EndFunction
		
		\Function{MOVConstructive}{$k$}
		\State $h(S) \coloneq \E_y\left[f(S, y, V_{c_*}^2)\right]$
		\State\Return \textsc{Greedy}($h$, $k$)
		\EndFunction
		
		\Function{MOVDestructive}{$k$}
		\State $h(S) \coloneq \E_y\left[f(S, y, V_{c_*}^1)\right]$
		\State\Return \textsc{Greedy}($h$, $k$)
		\EndFunction
		
	\end{algorithmic}
\end{algorithm}

\begin{algorithm}
	\caption{Algorithms for POV objective} \label{alg:pov}
	\begin{algorithmic}[1] 
		\Function{EnumerateThreshold}{$h$, $\Delta$, $k$}
		\For{$\beta = \Delta...n$}
		\State $h'(S) \coloneq \E_y\left[\min\left(\beta, h(S, y)\right)\right]$
		\State $S_\beta = $ \textsc{Greedy}$(h', k)$
		\EndFor		
		\State\Return $\arg\max_{S_\beta, \beta = \Delta....n} \text{Pr}_y\left[h(S_\beta, y) \geq \Delta\right]$
		\EndFunction
		
		\Function{POVConstructive}{$k$}
		\State //recall $m_C(S, y) = \E_y\left[f(S, y, V_{c_*}^2)\right]$ for 2-candidate case
		\State\Return \textsc{EnumerateThreshold}($m_C$, $\Delta_C$, $k$)
		\EndFunction
		
		\Function{POVDestructive}{$k$}
		\State\Return \textsc{EnumerateThreshold}($m_D$, $\Delta_D$, $k$)
		\EndFunction
		
	\end{algorithmic}
\end{algorithm}

\begin{theorem}
	In an election with two candidates, $\mov$ is a monotone submodular function. \label{theorem:movsubmodular}
\end{theorem}

\begin{proof}
	 We first fix a particular scenario $y$ and show that the function $f(\cdot, y, V_{c_*}^2)$ is submodular. This suffices to show that $\E_y[f(\cdot, y, V_{c_*}^2)]$ is submodular since a nonnegative linear combination of submodular functions remains submodular. Monotonicity is clear since adding additional seeds to $A$ can only make more nodes reachable. To show submodularity, we can write the marginal gain as
	
	\begin{align*}
	f(A \cup \{x\}, y, V_{c_*}^2) - f(A, y, V_{c_*}^2)  = \sum_{v \in V_{c_*}^2} (1 - \chi(v, A, y)) \chi(v,\{x\}, y).
	\end{align*} 
	
	Compare the above expression for a set $A$ and any $B \supseteq A$. For any single node $v$, $\chi(v, B, y) = 1$ whenever $\chi(v, A, y) = 1$. Hence, the term in the above summation for each node $v$ can only be smaller for $f(B \cup \{x\}, y, V_{c_*}^2) - f(B, y, V_{c_*}^2)$ than for $f(A \cup \{x\}, y, V_{c_*}^2) - f(A, y, V_{c_*}^2)$. We conclude that $f(A \cup \{x\}, y, V_{c_*}^2) - f(A, y, V_{c_*}^2) \geq f(B \cup \{x\}, y, V_{c_*}^2) - f(B, y, V_{c_*}^2)$ and submodularity now follows by taking the expectation over $y$. 
\end{proof}

Hence, we can apply a greedy algorithm to $\mov$ to obtain a $(1 - 1/e)$-approximation (\textsc{MOVConstructive} in Algorithm \ref{alg:mov}). Moreover, this ratio is tight since $\mov$ contains regular influence maximization as a special case when all nodes have $c_*$ in second place ($V_{c_*}^2 = V$). It is NP-hard to approximate influence maximization with ratio better than $1 - 1/e$ \cite{kempe2003maximizing}. Hence, two-candidate $\mov$ is computationally intractable with respect to exact optimization but high quality and efficient approximation algorithms exist. 

We now turn to the $\pov$ objective, where we maximize the \emph{probability} that $c_*$ wins the election. It is natural to think that submodularity may also carry over to this setting. However, this is not the case; we can provide a simple counterexample where $\pov$ violates submodularity. Consider $n$ isolated nodes, where $\frac{n}{2} - k + 1$ have $c_*$ as their first choice. Hence, to win the election it is necessary and sufficient to influence $k$ nodes, which can be accomplished by any choice of $k$ seeds from among those with $c_*$ as their second choice. Fix a seed set $B$ containing $k-1$ of these nodes and consider any $A \subset B$ (that is, $A$ is strictly smaller). We have $\pov(B) = \pov(A) = 0$. By adding a node $v \in V_{c_*}^2\setminus B$ to $B$, we have $\pov(B \cup \{v\}) - \pov(B) = 1$. However, since $|A| < k-1$, $\pov(A \cup \{v\}) = 0$ and hence $\pov(A \cup \{v\})- \pov(A) = 0$. This contradicts the definition of submodularity. Essentially, the $\pov$ objective displays a sharp threshold behavior, where additional seed nodes have no value until we are close to winning. This behavior in fact translates into the following strong hardness result:

\begin{theorem}
	It is NP-hard to compute an $\alpha$-approximation to the problem $\max_{|S| \leq k} \pov(S)$ for any $\alpha > 0$, even for two candidates and even when the instance is deterministic. \label{theorem:povhard}
\end{theorem}

\begin{proof}
	We consider a deterministic objective: the ICM with all propagation probabilities either 0 or 1. Without loss of generality, we have only a single scenario and will drop the dependence on $y$ in $f$. Suppose that we have an $\alpha$-approximation for $\pov$-maximization. We show how we can use this algorithm to optimally solve the influence maximization problem (i.e., maximizing $f(\cdot, V)$), which is known to NP-hard since it includes maximum coverage as a special case. Let $OPT_{IM}$ be the optimal value of the influence maximization problem and $OPT_{POV}(\Delta)$ be the optimal value for Problem 1 with the given threshold. Fix any $\Delta > 0$. If $\Delta \leq OPT_{IM}$, then there is a set $S$ with $\pov(S) = 1$ and hence $OPT_{POV}(\Delta) = 1$. Otherwise, there is no set with value $\Delta$ and $OPT_{POV}(\Delta) = 0$. Since the objective to the $\pov$ problem is either 0 or 1, any $\alpha$-approximation algorithm for it must return 1 whenever $OPT_{POV}(\Delta) = 1$. Now, we can just enumerate over $\Delta = 1...n$, where $n$ is the number of nodes in the graph. At each value of $\Delta$, we ask the $\alpha$-approximation algorithm to solve the $\pov$ maximization problem with that value of $\Delta$. We return the solution corresponding to the highest value of $\Delta$ for which we can find a set with $f(S, V) \geq \Delta$. By the above, this set is an optimal solution to the influence maximization problem. 
\end{proof}

We remark that since this hardness result has broader implications. Recall that $\pov(S) = \pr_y\left[f(S, y, V_{c_*}^2) \geq \Delta\right]$ where by Theorem \ref{theorem:movsubmodular}, $f(S, y, V_{c_*}^2)$ is a submodular function. Therefore, the inapproximability result in Theorem \ref{theorem:povhard} shows that it is in general hard to approximate the problem of of maximizing the probability that a submodular function exceeds a given threshold value. This is a natural objective in other domains, e.g. for a risk-averse decision maker who wants to control the probability of a bad outcome.

We pair this hardness result with a positive algorithmic result regarding bicriteria approximations. A bicriteria approximation algorithm gives up solution quality in more than one dimension, and is of interest when hardness results preclude the usual notion of approximation (as for our problem). We provide an algorithm which has a solution quality guarantee whenever the election is winnable by a ``large margin". That is, there is a seed set with high probability of greatly exceeding $\Delta$ votes. Our algorithm will attempt to maximize the probability of exceeding exactly $\Delta$ votes, but has a guarantee relative to the optimal value for threshold $\frac{1}{\alpha}\Delta$ for some $\alpha < 1$. That is, it is only compared to the optimal value of a harder problem.
 
Our algorithm is a greedy strategy based on the surrogate function $h(S) =  \E_y\left[\min\{\beta, f(S, y, V_{c_*}^2)\}\right]$, where $\beta$ is a chosen threshold value. The intuition is to replace the sharp discontinuity of the original $\pov$ objective by a surrogate which interpolates smoothly up to the threshold $\beta$. However, we do not give any ``credit" for nodes reached beyond $\beta$ since (unlike in the $\mov$ case) we only care about crossing the threshold. It is easy to see that the minimum of a submodular function and a constant is itself submodular \cite{krause2008robust}. Hence, $h$ is submodular and amenable to greedy optimization. \textsc{POVConstructive} (Algorithm \ref{alg:pov}) iterates over a series of possible thresholds for $h$, optimizes each one greedily, and outputs the best of the resulting seed sets. Specifically, it tries every value of $\beta$ from $\Delta...n$. For each $\beta$, it finds a seed set $S_\beta$ by greedily optimizing $\E_y\left[\min\{\beta, f(S, y, V_{c_*}^2)\}\right]$ (Algorithm \ref{alg:pov}, Lines 3-4). Then, it outputs the $S_\beta$ which maximizes $\pr_y\left[f(S_\beta, y, V_{c_*}^2) \geq \Delta\right]$, i.e., the one which has the best probability of exceeding the true objective (Line 5). 

The reason that we need to enumerate over values for $\beta$, instead of just using the true threshold $\Delta$, is that optimizing the surrogate $h$ might result in a solution which has value below $\beta$ in every scenario. However, we can show that if $OPT(\beta)$ is high, there must be many scenarios where $S_\beta$ has value close to $\beta$ (a notion formalized in our proof). Hence, if we try a  sufficiently large $\beta > \Delta$ and $OPT(\beta)$ is still high, there must be many scenarios with value at least $\Delta$. This is formalized in the following guarantee:

\begin{theorem}
	In an election with two candidates, \textsc{POVConstructive} produces a solution $S$ such that 
	$$
	\pr_y[f(S, y, V_{c_*}^2) \geq \Delta] \geq \max_{0 < \alpha < 1} \frac{\left(1 - \frac{1}{e}\right) OPT_{POV}\left(\frac{1}{\alpha}\Delta\right) - \alpha}{1 - \alpha}
	$$ \label{theorem:povbicriteria}
\end{theorem}

\begin{proof}
	
	\textsc{POVConstructive} enumerates over values of $\alpha$ by trying thresholds $\beta = \Delta...n$. Fix a specific $\beta$ and set $\alpha = \frac{\Delta}{\beta}$. We will prove that $\pr_y[f(S_\beta, y, V_{c_*}^2) \geq \Delta] \geq \frac{\left(1 - \frac{1}{e}\right) OPT_{POV}\left(\frac{1}{\alpha}\Delta\right) - \alpha}{1 - \alpha}$. This suffices the prove the theorem because we output the best of the $S_\beta$. A minor point is that the theorem takes the max over $0 < \alpha <1$, while we try only the discrete points $\alpha = \frac{\Delta}{\Delta}, \frac{\Delta}{\Delta+1}, ..., \frac{\Delta}{n}$. However, these are equivalent because $f(S, y, V_{c_*}^2)$ is always integral.
	
	 We divide the set of scenarios into those where the $S_\beta$ has value at least $\alpha \beta$ and those where it has less value. Let $A = \{y: f(S_\beta, y, V_{c_*}^2) \geq \alpha \beta\}$ and $B = \{y: f(S_\beta, y, V_{c_*}^2) < \alpha \beta\}$. We have
	
	\begin{align*}
	&\frac{1}{m}\sum_{y\in A} \min\{\beta, f(S_\beta, y, V_{c_*}^2)\} + \frac{1}{m}\sum_{y \in B} \min\{\beta, f(S_\beta, y, V_{c_*}^2)\}\\ &\geq \left(1 - \frac{1}{e}\right) \max_{|S| \leq k} h(S) \,\, \geq \left(1 - \frac{1}{e}\right) \beta m \cdot OPT_{POV}(\beta)
	\end{align*}
	
	where the first inequality follows from submodularity and the second follows since the solution attaining value $OPT_{POV}(\beta)$ for the POV maximization problem is a feasible solution to the problem $\max_{|S| \leq k}h(S)$ which has value at least $\beta m \cdot OPT_{POV}(\beta)$. We are interested in the minimum possible size of $A$ given that the total value is lower bounded as above. By inspection, $|A|$ is minimized when $\min\{\beta, f(S, y)\} = \beta$ for each $y \in A$ and $f(S, y) = \alpha \beta$ for each $y \not\in A$. In this case, we have 
	
	\begin{align*}
	\frac{|A|}{m} + \alpha\left(1 - \frac{|A|}{m}\right) \geq \left(1 - \frac{1}{e}\right) OPT_{POV}(\beta)
	\end{align*}  
	
	and hence
	
	\begin{align*}
	&\frac{|A|}{m} = \frac{1}{m}\sum_{i = 1}^m 1[f(S_\beta, y, V_{c_*}^2) \geq \alpha \beta] = \pr_y\left[f(S_\beta, y, V_{c_*}^2) \geq \Delta\right]\\
	&\geq \frac{\left(1 - \frac{1}{e}\right) OPT_{POV}\left(\frac{1}{\alpha}\Delta\right) - \alpha}{1 - \alpha}
	\end{align*}
	
	which completes the proof.
\end{proof}

Theorem \ref{theorem:povbicriteria} in fact applies to the general problem of maximizing the probability that a submodular function exceeds a threshold value (complementing our hardness result in Theorem \ref{theorem:povhard}). As discussed above, this may be of interest independently of election control.


\section{Multiple candidates}

We now consider election control with more than two candidates. There is a target candidate $c_*$ and other candidates $c_1...c_\ell$. Note that constructive and destructive control are distinct in this setting. We will give algorithms for both cases for both the MOV and POV objectives.

The problem becomes significantly harder in the multicandidate setting because we must now reason simultaneously about several objectives -- whether each alternate candidate $c_i$ will accumulate more votes that $c_*$. We demonstrate that, up to the loss of a constant in the approximation ratio, it suffices to concentrate only on the number of votes gained or lost by $c_*$ (not the margin against each $c_i$ individually). This concept yields (bicriteria) approximation algorithms for each setting along the lines of the two-candidate case.

We start out with the $\mov_C$ objective (constructive control for the margin of victory), since the idea is simpler to illustrate in this case. The basic intuition is that the change in margin between candidate $c_i$ and $c_*$ can be re-expressed as follows:

\begin{align*}
g_C(S, y, c_i) &= \sum_{v \in V_{c_*}^2 \setminus V_{c_i}^1} \chi(v, S, y) + 2\sum_{v \in V_{c_*}^2 \cap V_{c_i}^1} \chi(v, S, y)\\
&= f\left(S, y, V_{c_*}^2\right) + f\left(S, y, V_{c_*}^2 \cap V_{c_i}^1\right)
\end{align*}

Now, we can express the final margin in scenario $y$ as


\begin{align*}
m_C(S, y) = f\left(S, y, V_{c_*}^2\right) + \min_{c_j} \left(f\left(S, y, V_{c_*}^2 \cap V_{c_j}^1\right) + \max_{c_i} |V_{c_i}^1| - \left|V_{c_j}^1\right|\right)
\end{align*}

where the first term is common to all candidates and reflects the total number of voters who switch to $c_*$ and the min term selects the $c_{j}$ who has the most remaining votes. In general, this second term can be very difficult to approximate because it is the minimum of submodular functions, which is not in general submodular (or even approximable \cite{krause2008robust}). We might hope that there is some special structure to the election control problem, but this is not the case:  

%
%

\begin{theorem}
	For any $\epsilon > 0$, it is NP-hard to compute any $\Omega\left(\frac{1}{n^{1 - \epsilon}}\right)$-approximation to the problem 
	\begin{align*}
	\max_{|S| \leq k}\min_{c_j} \left(f\left(S, y, V_{c_*}^2 \cap V_{c_i}^1\right) + \max_{c_i} |V_{c_i}^1| - |V_{c_j}^1|\right)
	\end{align*}
\end{theorem}

\begin{proof}
	We will consider instances where all of the $c_j$ start with an equal number of votes and so $\max_{c_i} |V_{c_i}^1| - |V_{c_j}^1| = 0$ for all $c_j$. Thus, the problem is just $\max_{|S| \leq k}\min_{c_j} f\left(S, y, V_{c_*}^2 \cap V_{c_i}^1\right)$. We reduce from the robust influence maximization (RIM) problem \cite{heRobust2016,chen2016robust}. In RIM, we are given a set of objectives $f_1...f_r$, each of which represent expected influence spread in an instance of the independent cascade model on a common underlying graph $G$. He and Kempe \cite{heRobust2016} show that it is NP-hard to compute $\Omega\left(\frac{1}{n^{1 - \epsilon}}\right)$-approximation to the problem $\max_{|S| \leq k} \min_{i = 1...r} f_i(S)$. Their proof holds when each $f_i$ is deterministic (assigns probability 0 or 1 to each edge), so we will assume that the instance is in this form. Let $G_i$ be a graph in which each edge of $G$ assigned probability 0 by $f_i$ has been removed. We create a graph $G'$ as follows. $G'$ contains each $G_i$ as a disconnected subgraph. For every $v \in G$, we add a vertex $v'$ to $G'$. $v'$ has an outgoing edge to the copy of node $v$ in each of the $G_i$ subgraphs. Each such edge has propagation probability 1. There is a target candidate $c_*$ and $r$ additional candidates $c_1...c_r$. Each of the $v'$ nodes that were added has $c_*$ as their first choice. Each node in subgraph $G_i$ has $c_{i}$ as their first choice and $c_*$ as their second choice. 
	
	Suppose that we have an $\alpha$-approximation algorithm for our problem for some $\alpha = \Omega\left(\frac{1}{n^{1 - \epsilon}}\right)$. Without loss of generality, we will assume that this algorithm only selects nodes from the $v'$ (since if a seed set contains the copy of $v$ in any subgraph, we can only obtain greater influence spread by exchanging it for $v'$). Note that, for any such set of seed nodes, $f_i(S) = f(S, V_{c_*}^2 \cap V_{c_i}^1)$ Thus, if $S$ is an $\alpha$-approximate solution for our problem, it is also an $\Omega\left(\frac{1}{n^{1 - \epsilon}}\right)$-approximate solution to the RIM problem. 
\end{proof}

Therefore, we should not hope for any algorithm which can closely approximate the entirety of the objective; the $\min$ component is too difficult to handle. However, we can leverage the fact that the first term, $f(S, y, V^2_{c_*})$, is easy to optimize because it is just a submodular function. Hence, the objective is the sum of an easy term and a hard term. Importantly, we can show that optimizing just the easy term (which is what \textsc{MOVConstructive} does) is sufficient to obtain a constant factor approximation. 

\begin{theorem}
	\textsc{MOVConstructive} obtains a $\frac{1}{3}\left(1 - \frac{1}{e}\right)$- approximation to the $\mov_C$ problem with any number of candidates. \label{theorem:mov-c}
\end{theorem}
\begin{proof}
 
 Let $c(S, y) = \arg\min_{c_i} f(V_{c_*}^2 \cap V_{c_i}^1) - |V_{c_i}^1|$ be the candidate achieving the minimum in the definition of $m_C$. Let $S^*$ be an optimal seed set. Note that for all scenarios $y$, seed sets $S$, and candidates $c_i$, $f(S, y, V_{c_*}^2) \geq f(S, y, V_{c_*}^2 \cap V_{c_i}^1)$. Hence, we have
 
%
%
%

 \begin{align*}
\E_y\Big[f\left(S^*, y, V_{c_*}^2\right)\Big]  &\geq \frac{1}{3} \E_y\Big[f\left(S^*, y, V_{c_*}^1\right) + f\left(S^*, y, V_{c_*}^1 \cap V_{c(S^*, y)}^1\right) \\&\quad\quad\quad\quad\quad\quad\quad\quad+ f\left(S^*, y, V_{c_*}^1 \cap V_{c(S, y)}^1\right)\Big]\\
\end{align*}

Note that $\E_y[f(\cdot, y, V_{c_*}^2)]$ is a monotone submodular function, which \textsc{MOVConstructive} greedily maximizes. Let $S$ be the resulting seed set. We have 

\begin{align*}
&\E_y\Big[f\left(S, y, V_{c_*}^2\right) + f\left(S, y, V_{c_*}^2 \cap V_{c(S, y)}^1\right)\Big] \\&\geq \E_y\Big[f\left(S, y, V_{c_*}^2\right)\Big]\\
&\geq \frac{1}{3}\left(1 - \frac{1}{e}\right) \E_y\Big[f\left(S^*, y, V_{c_*}^2\right) +  f\left(S^*, y, V_{c_*}^2 \cap V_{c(S^*, y)}^1\right) + f\left(S^*, y, V_{c_*}^2 \cap V_{c(S, y)}^1\right)\Big]
\end{align*}

which allows us to bound the margin of victory relative to $S^*$ as 
\allowdisplaybreaks
\begin{align*}
\mov_C(S) &= \E_y\Big[f\left(S, y, V_{c_*}^2\right) + \min_{c_j}f\left(S, y, V_{c_*}^2 \cap V_{c_j}^1\right) + \max_{c_i} |V_{c_i}^1| - |V_{c_j}^1|\Big]\\
&=\E_y\Big[f\left(S, y, V_{c_*}^2\right) + f\left(S, y, V_{c_*}^2 \cap V_{c(S, y)}^1\right)\Big] + \max_{c_i} |V_{c_i}^1| - \E_y\Big[|V_{c(S, y)}^1|\Big]\\
&\geq \frac{1}{3}\left(1 - \frac{1}{e}\right) \E_y\Big[f\left(S^*, y, V_{c_*}^1\right) + f\left(S^*, y, V_{c_*}^2 \cap V_{c(S^*, y)}^1\right) + f\left(S^*, y, V_{c_*}^2 \cap V_{c(S, y)}^1\right)\Big] + \max_{c_i} |V_{c_i}^1| - \E_y\Big[|V_{c(S, y)}^1|\Big]\\
\end{align*}

and some additional algebra (deferred to the appendix) yields

\begin{align*}
&\mov_C(S) \geq \frac{1}{3}\left(1 - \frac{1}{e}\right)\Bigg(\mov_C(S^*) + \E_y\Big[f\left(S^*, y, V_{c_*}^2 \cap V_{c(S, y)}^1\right) + |V_{c(S^*, y)}^1| - |V_{c(S, y)}^1|\Big]\Bigg).
\end{align*}

Now by definition of $c(S^*, y)$, $f(S^*, y, V_{c_*}^2 \cap V_{c(S^*, y)}^1) - |V_{c(S^*, y)}^1| \leq f(S^*, y, V_{c_*}^2 \cap V_{c(S, y)}^1) - |V_{c(S, y)}^1|$ and so 

\begin{align*}
|V_{c(S^*, y)}^1|- |V_{c(S, y)}^1| \geq f\left(S^*, y, V_{c_*}^2 \cap V_{c(S^*, y)}^1\right) - f\left(S^*, y, V_{c_*}^2 \cap V_{c(S, y)}^1\right)
\end{align*}

This yields

\begin{align*}
&\mov_C(S) \geq \frac{1}{3}\left(1 - \frac{1}{e}\right)\Bigg(\mov_C(S^*) + \E_y\Big[f\left(S^*, y, V_{c_*}^2 \cap V_{c(S, y)}^1\right) \\&\quad\quad\quad\quad\quad+f\left(S^*, y, V_{c_*}^2 \cap V_{c(S^*, y)}^1\right) - f\left(S^*, y, V_{c_*}^2 \cap V_{c(S, y)}^1\right)\Big]\Bigg)\\
&=\frac{1}{3}\left(1 - \frac{1}{e}\right)\Bigg(\mov_C(S^*) + \E_y\Big[f\left(S^*, y, V_{c_*}^2 \cap V_{c(S^*, y)}^1\right)\Big]\Bigg)\\
&\geq \frac{1}{3}\left(1 - \frac{1}{e}\right)\mov_C(S^*).
\end{align*}

\end{proof}

We also have a corresponding result for the destructive control case. Here, we can rewrite the change in margin as $m_D(S, y) = f(S, y, V_{c_*}^1) + \max_{c_i} \left(f(S, y, V_{c_*}^1 \cap V_{c_i}^2) + |V_{c_i}^1| - \max_{c_j} |V_{c_j}^1|\right)$. \textsc{MOVDestructive} greedily optimizes the submodular function $\E_y\left[f(S, y, V_{c_*}^1)\right]$, which we show is a good surrogate for $\E_y\left[m_D(S, y)\right]$. 

\begin{theorem}
	\textsc{MOVDestructive} obtains a $\frac{1}{2}\left(1 - \frac{1}{e}\right)$-approximation to the multicandidate $\mov_D$ problem. \label{theorem:mov-d}
\end{theorem}

The proof, which is similar to that of Theorem \ref{theorem:mov-c}, can be found in the appendix.

Now, we extend these ideas to obtain similar guarantees for the $\pov_C$ and $\pov_D$ objectives. Starting with $\pov_C$, recall that our objective is to maximize $\pr_y\left[m_C(S, y) \geq \Delta_C\right]$, the probability that the change in margin exceeds the number of votes needed to win. We will prove a guarantee for the same algorithm \textsc{POVConstructive} as from the two-candidate case. Recall that \textsc{POVConstructive} optimizes the surrogate $\E_y\left[\min\left(\beta, f\left(S, y, V_{c_*}^2\right)\right)\right]$, enumerating over possible values of the threshold $\beta$. We have the following bicriteria approximation guarantee:

\begin{theorem}
	Let $OPT(\Delta)$ denote the optimal value of the problem $\max_{|S| \leq k} \pr_y\left[m_C(S, y) \geq \Delta\right]$. Let $S$ be the set produced by \textsc{POVConstructive}. We have
	\begin{align*}
	\pov_C(S) \geq \max_{0 < \alpha < 1} \frac{\frac{e-1}{3e-1} OPT\left(\frac{1}{\alpha} \Delta_C\right) - \alpha}{1 - \alpha}
	\end{align*} \label{theorem:povc}
\end{theorem}

The proof can be found in the appendix. The main difference from the two candidate case is that we do not directly optimize $\frac{1}{m}\sum_y \min\left(\beta, m_C(S, y)\right)$ since it may no longer be submodular. Instead, we greedily optimize the submodular surrogate function $\frac{1}{m}\sum_y \min\left(\beta, f(S, y, V_{c_*}^2)\right)$ and show that this surrogate approximates $\frac{1}{m}\sum_y \min\left(\beta, m_C(S, y)\right)$. From there, the same argument as in Theorem \ref{theorem:povbicriteria} extends to the multicandidate case. Analogous reasoning also yields a bicriteria guarantee for destructive control:

\begin{theorem}
	Let $OPT(\Delta)$ denote the optimal value of the problem $\max_{|S| \leq k} \pr_y\left[m_D(S, y) \geq \Delta\right]$. Let $S$ be the set produced by \textsc{POVDestructive}. We have
	\begin{align*}
	\pov_D(S) \geq \max_{0 < \alpha < 1} \frac{\frac{e-1}{3e-1} OPT(\frac{1}{\alpha} \Delta) - \alpha}{1 - \alpha}
	\end{align*} \label{theorem:povd}
\end{theorem}

\section{Exact solutions via mixed-integer programming}

Thus far, we have only considered approximation algorithms for election control, motivated by computational hardness results for exact optimization. Now we give mixed-integer linear programming (MILP) formulations to find exact solutions. This serves two purposes. First, it allows us to study the effectiveness of election control for problem instances with are within the range of state of the art MILP solvers. Second, we can determine the empirical effectiveness of the approximation algorithms proposed in earlier sections. 

There are two principal difficulties in obtaining MILP formulations. First, the objective is stochastic, ranging over an exponential number of scenarios. Second, even for a single fixed scenario, the number of nodes reached by a seed set is a nonlinear function. 


We first show how to linearize the problem when have only a single scenario $y$. Recall that $y$ corresponds to a sampled graph $G_y$, where every edge $e$ is removed independently with probability $1 - p_e$. Our MILP will have a binary variable $s_v \in \{0,1\}$ for each node $v \in V$, where $s_v = 1$ indicates that $v$ is a seed node. We will maximize an objective over all $s_v \in \{0,1\}^{|V|}$ which satisfy $\sum_{v \in V} s_v \leq k$ (at most $k$ nodes are seeded). The challenge is to embed the nonlinear objective into the constraints of the MILP. Let $x^y_v, v \in V$ be a binary variable indicating whether $v$ is influenced in scenario $y$. We must constrain $x^y_v$ to be 1 only if $v$ truly is reachable in $G_y$ from some node with $s_v = 1$. To accomplish this, let $R(v, y)$ be the set of nodes which have a directed path to $v$ in scenario $y$. $R(v, y)$ does not depend on the decision variables $s$ and can easily precomputed. Using this set, we constrain the $x$ variables as:

\begin{align*}
x^y_v \leq \sum_{u \in R(v, y)} s_u \quad \forall v \in V.
\end{align*}

Now we deal with stochasticity in the objective using sample average approximation. We first sample a set of scenarios $G_{y_i}, i = 1...r$, maintaining a separate copy $x_v^{y_i}$ for each sampled scenario. Finally, we average over the variables in each scenario to obtain the final objective. 

\subsection{Formulations}

Using these components, we now give concrete formulations for each of the problem instances that we consider. We will assume that scenarios $y_1...y_r$ have been sampled, where $r$ is a tunable parameter trading off computational cost and sampling error. 

\subsubsection{Constructive control}

We create a variable $g_C(y_i, c_j)$ for each scenario $y_i$ and candidate $c_j$ which represents the change in the margin between $c_j$ and $c_*$ in scenario $y_i$. Using these variables, we set a variable $m_C(y_i)$ for each scenario $y_i$ which represents the overall change in margin. These variables are set using the constraints

%

\begin{align*}
g_C(y_i, c_j) \leq \sum_{v \in V_{c_*}^2} x_v^{y_i} + \sum_{v \in V_{c_*}^2 \cap V_{c_j}^1} x_v^{y_i}\\
m_C(y_i) \leq g_C(y_i, c_j) + \max_{c_i} |V_{c_i}^1| - |V_{c_j}^1| \forall i,j
\end{align*}

Using these variables, we have the following MILP to maximize the MOV:

\begin{align*}
&\max_{s, x, g_C, m_C} \frac{1}{r}\sum_{i = 1}^r m_C(y_i)\\
&\sum_{v \in V} s_v \leq k
\end{align*}

The next formulation maximizes the POV:

\begin{align*}
&\max_{s, x, g_C, m_C, u} \frac{1}{r}\sum_{i = 1}^r u_i\\
&-M(1 - u_i) + \max_{c_i} |V_{c_i}^1| - |V_{c_*}^1| + 1 - m_C(y_i) \leq 0\\
&\sum_{v \in V} s_v \leq k, \quad u_i \in \{0,1\} \quad i = 1...r
\end{align*}

Here, $u_i$ is a binary variable representing whether $c_*$ wins the election in scenario $y_i$, while $M$ is a large number.

\subsubsection{Destructive control}

Now, we use an analogous set of constraints to set variables $g_D(y_i, c_j)$ and $m_D(y_i)$:

\begin{align*}
&g_D(y_i, c_j) \leq \sum_{v \in V_{c_*}^1} x_v^{y_i}+ \sum_{v \in V_{c_*}^1 \cap V_{c_j}^2} x_v^{y_i}\\
&-M(1 - z_i^j) + m_D(y_i) - \left(g_D(y_i, c_j) + |V_{c_j}^1| - \max_k |V_{c_k}^1|\right) \leq 0 \quad \forall i,j\\
&\sum_j z_i^j \geq 1 \,\,\forall i, \quad z_i^j \in \{0,1\} \,\,\forall i, j
\end{align*}

The second and third constraints use a new set of binary variables $z_i^j$, where $z_i^j = 1$ indicates that in scenario $y_i$, $m_D(y_i)$ is at most the change in margin between $c_j$ and $c_*$. The constraint $\sum_j z_i^j \geq 1$ requires that $m_D(y_i)$ must be bounded by one such value, and so can be at most the maximum margin. With these variables in place, the MOV and POV MILPs are analogous to those for constructive control. 

\section{Experiments}
\newcommand{\ra}[1]{\renewcommand{\arraystretch}{#1}}
\begin{table*}\centering
	\ra{1.3}
	\caption{Percent of MILP value obtained by approximation algorithm.}\label{table:mov-approx}
	\begin{tabular}{@{}rrrrcrrrcrrrccrr@{}}\toprule
		& \multicolumn{3}{c}{netscience} & \phantom{abc}& \multicolumn{3}{c}{facebook} &
		\phantom{abc} & \multicolumn{3}{c}{polblogs} & \phantom{abc} & \multicolumn{3}{c}{irvine}\\
		\cmidrule{2-4} \cmidrule{6-8} \cmidrule{10-12} \cmidrule{14-16}
		$k = $ & $25$ & $50$ & $100$ && $25$ & $50$ & $100$ && $25$ & $50$ & $100$ && $25$ & $50$ & $100$\\ \midrule
		Constructive\\
		$|C| = 2$ &99.5 & 99.3 & 100.  && 100 & 100 & 100  && 100 & 100 & 99.4  && 100 & 100 & 99.4 \\
		$|C| = 5$ &82.8 & 90.1 & 91.5  && 90.8 & 90.9 & 90.9  && 97.4 & 97.8 & 99.5  && 97.7 & 95.4 & 96.8 \\
		$|C| = 10$ &80.9 & 89.1 & 98.3  && 80.9 & 83.7 & 88.6  && 98.7 & 99.2 & 99.6  && 96.9 & 97.4 & 98.9 \\
		Destructive\\
		$|C| = 2$ &99.9 & 99.6 & 99.9  && 100 & 100 & 100  && 100 & 100 & 99.4  && 100 & 99.8 & 98.8 \\
		$|C| = 5$ &73.8 & 73.2 & 83.3  && 87.7 & 79.7 & 81.6  && 100 & 97.8 & 99.3  && 100 & 99.0 & 100 \\
		$|C| = 10$ &75.9 & 87.2 & 97.2  && 81.8 & 85.0 & 89.0  && 98.9 & 99.3 & 99.6  && 100 & 97.3 & 99.0 \\
		\bottomrule
	\end{tabular}
	
\end{table*}

We now present experimental results comparing the performance of our approximation algorithms to the solutions found via mixed integer programming. We show results on four datasets. First, \emph{netscience}, a collaboration network of researchers in network science, with 1461 nodes \cite{newman}. Second, \emph{facebook}, the subgraph centered on 10 Facebook users, with 2888 nodes \cite{fb}. Third, \emph{polblogs}, a network of links between political blogs, with 1224 nodes \cite{newman}. Fourth, \emph{irvine}, a graph representing instant messages exchanged between students at U.C.\ Irvine, with 1889 nodes \cite{fb}. We select these datasets because they represent the kinds of social and communication networks on which political messages (such as fake news) spread. We also note that our approximation algorithms can easily be scaled to much larger networks since we can apply the same techniques developed in the influence maximization literature \cite{tang2014influence,cohen2014sketch}. However, our focus here is to characterize the performance of our algorithms in comparison to the optimal solution, so we select datasets which are feasible for mixed integer programming. For each network, we randomly generated 30 sets of voter preferences.

We start out with the MOV objective. Table \ref{table:mov-approx} shows the percentage of the MILP's value which is obtained by our approximation algorithms (\textsc{MOVConstructive} and \textsc{MOVDestructive} respectively), averaging over the 30 instances on each network with propagation probability $p = 0.1$. We vary the number of seed nodes $k$ and the number of candidates $|C|$. We see that the approximation algorithms perform well across all settings, obtaining expected change in margin at least 73\% of that of the MILP. The approximation algorithms fare particularly well for 2-candidate elections, obtaining nearly 100\% of the optimal value on all networks. The empirical approximation ratio degrades as the number of candidates grows, particularly when the budget $k$ is small. We conclude that our approximation algorithms are highly effective for election control via the MOV in both constructive and destructive control, particularly in realistic settings with a moderate number of candidates.

\begin{figure}
	\centering
	\includegraphics[width=2in]{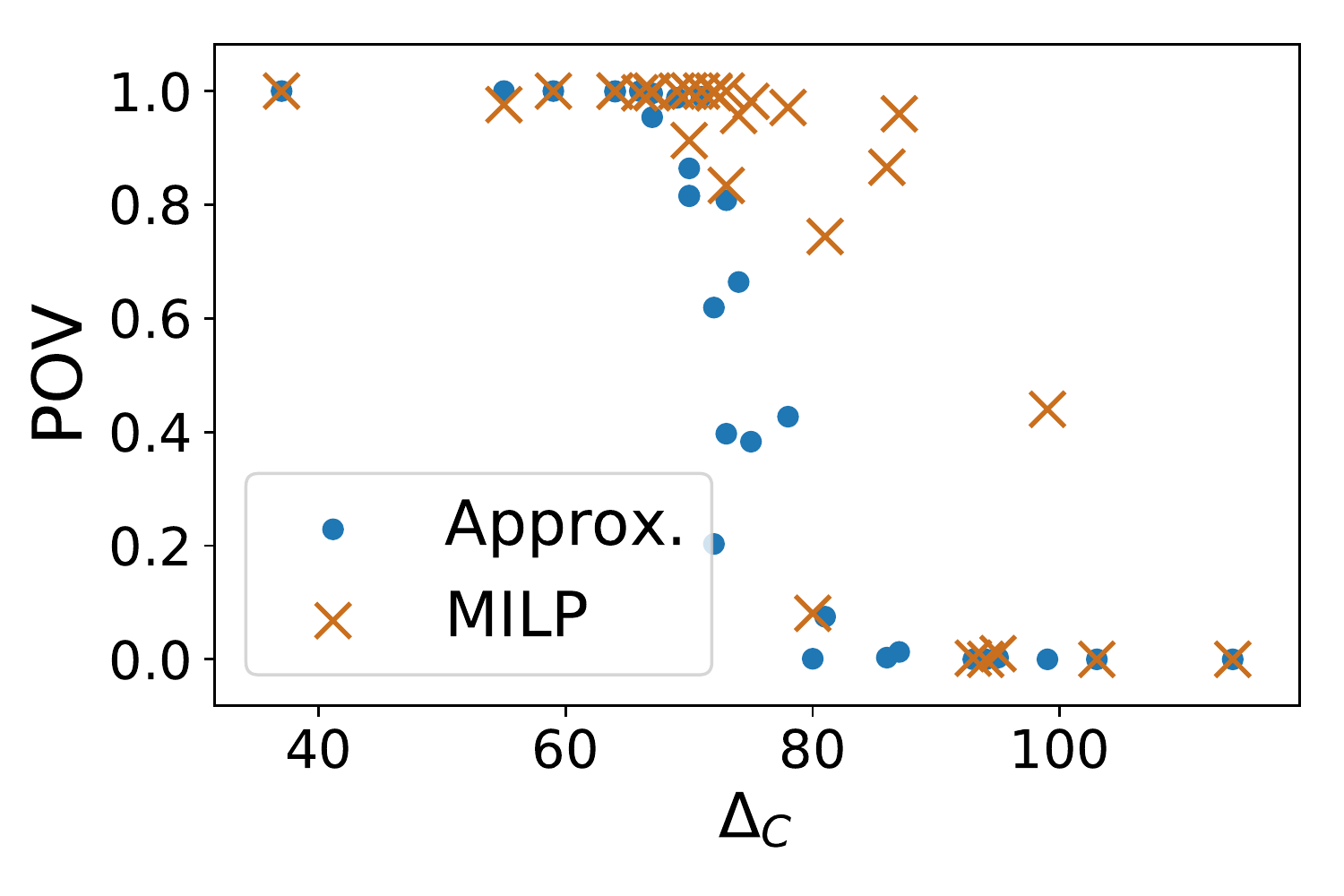}
	\includegraphics[width=2in]{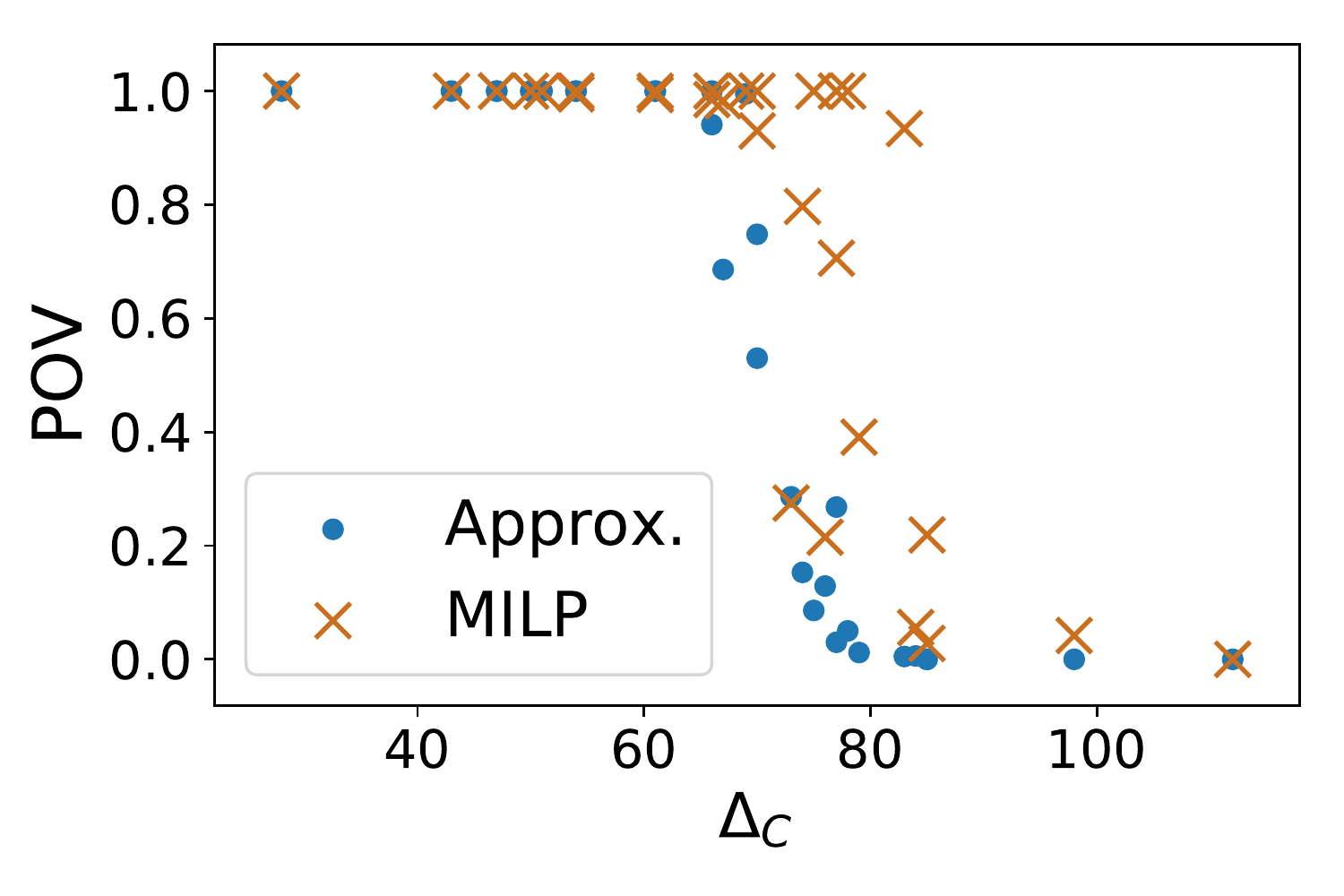}
	
	\includegraphics[width=2in]{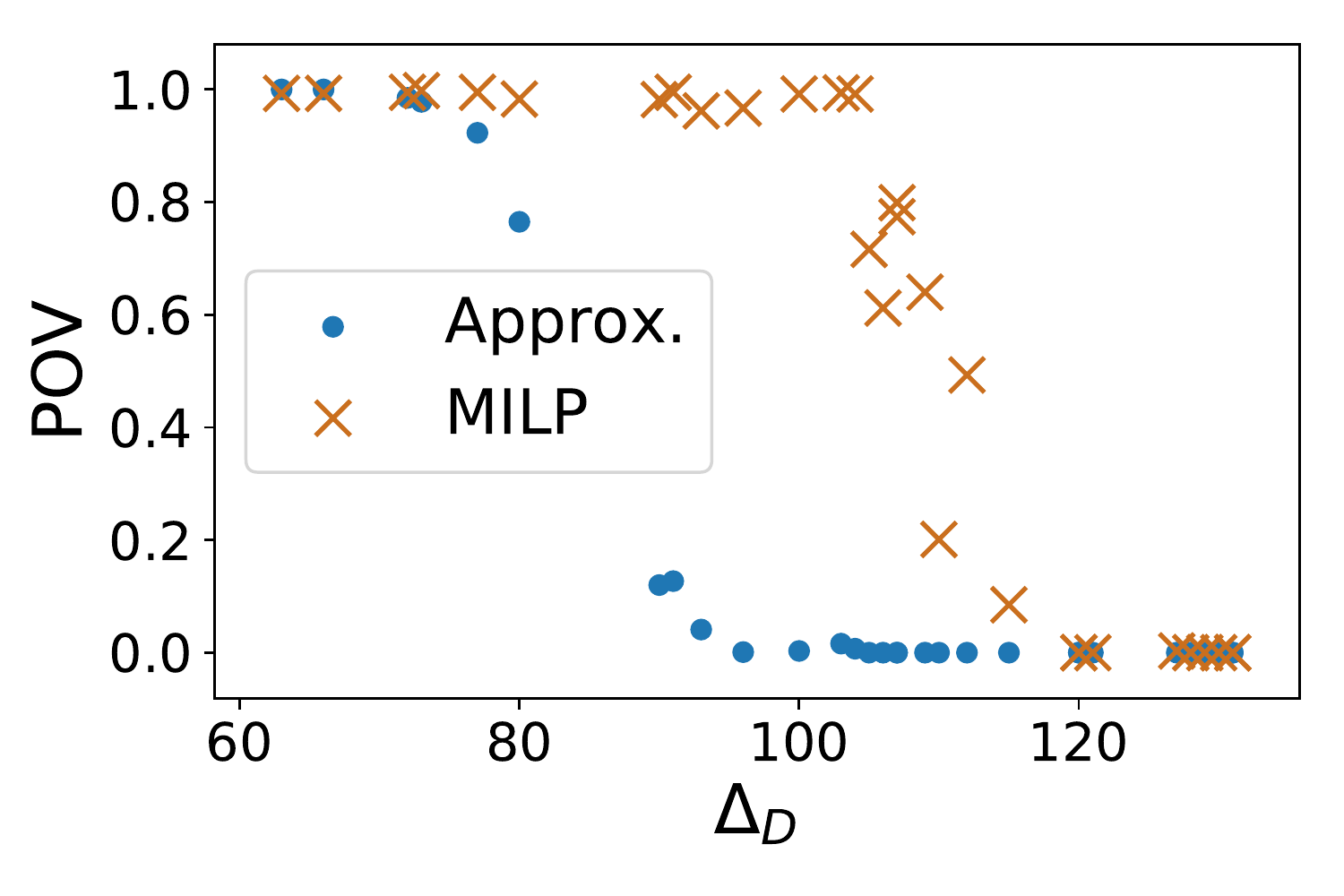}
	\includegraphics[width=2in]{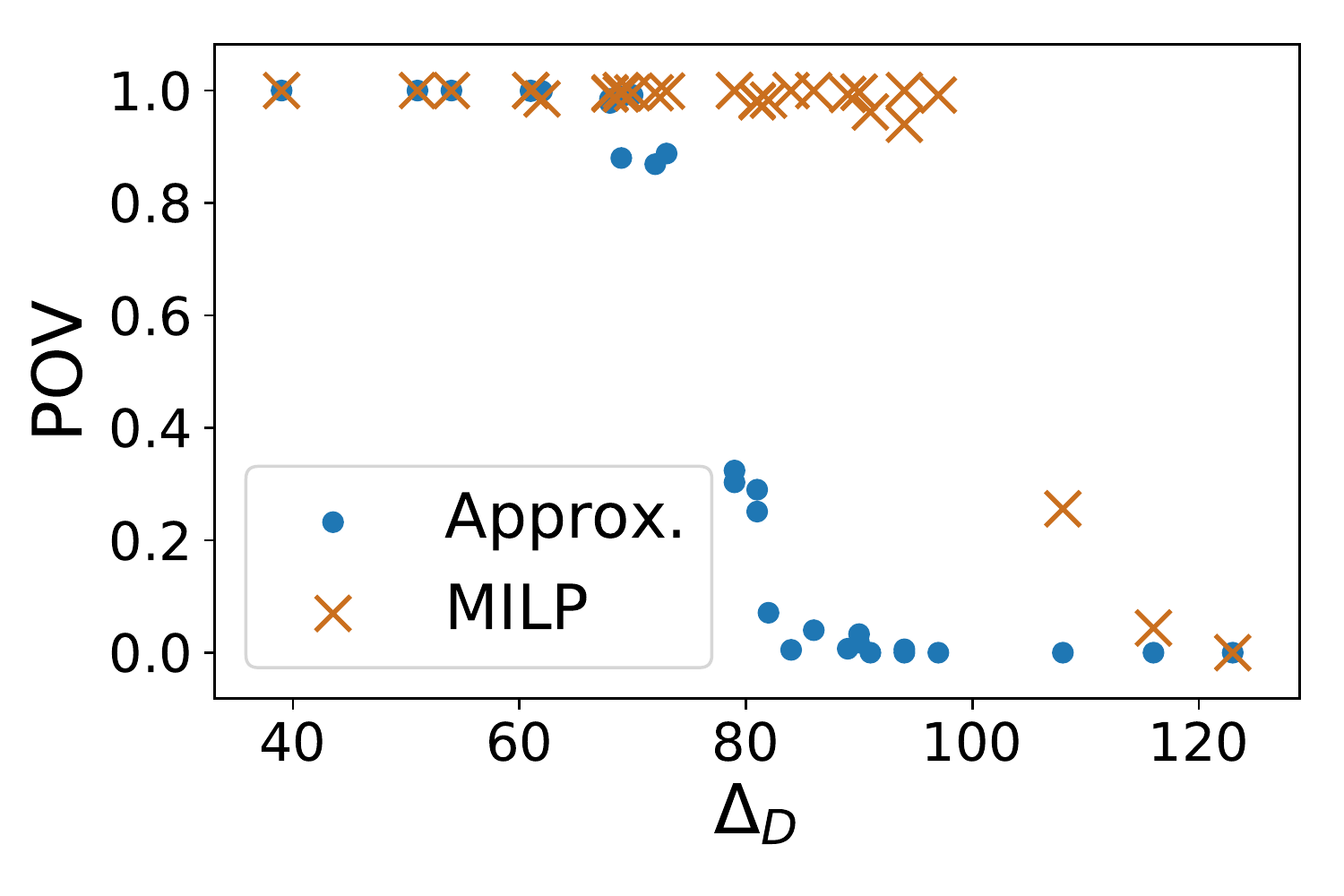}
	\caption{Probability of victory. Top: constructive. Bottom: destructive. Left: netscience. Right: polblogs.}\label{fig:pov}
\end{figure}

We now turn to the POV objectives. We show results for $k = 50, |C| = 5$, comparing our bicriteria approximation algorithms \textsc{POVConstructive} and \textsc{POVDestructive} to the corresponding MILP formulations. To keep the experiments timely, we ran each approximation algorithm for 150 random values of the threshold $\beta$ instead of enumerating over all (empirically, this resulted in very similar solution quality).  Figure \ref{fig:pov} shows the results on netscience and polblogs for constructive and destructive control. Results for facebook and irvine can be found in the appendix. The $x$ axis shows the starting margin ($\Delta_C$ or $\Delta_D$) in each randomly generated instance while the $y$ axis shows the probability of victory obtained. The instances fall into three groups. First, when the margin is small, both the approximation algorithms and the MILP have a high POV. Second, when the margin is large, both have a small POV. Third are intermediate points where the approximation algorithm and MILP strongly diverge. Averaged over all random instances, the approximation algorithm obtains 40-60\% of the MILP's value (depending on the network). However, there are instances among the intermediate cases where, e.g., the approximation algorithm obtains a POV of 0.1\%, but the MILP finds a solution with POV 99\%. We conclude that election control for the POV objective can be very computationally difficult in narrowly winnable elections. This dovetails with our theoretical results, which show that the problem is inapproximable in general, but bicriteria guarantees are possible when the election is winnable by a large margin.

\section{Conclusion}

Fake news and other targeted misinformation are an increasingly prevalent way of interfering with democratic elections. We introduce and study the problem of election control through social influence, providing algorithms and hardness results for maximizing both the margin and probability of victory for an attacker in both constructive and destructive control.  Our results indicate that social influence is a salient threat to election integrity, particularly in the MOV case where we provide high-quality approximation algorithms. Maximizing the probability of victory is manageable in easier instances, but difficult both theoretically and empirically in narrow races.
\fontsize{9.5}{10.5}
\selectfont
\bibliographystyle{plain}  
\bibliography{elections_bib}  
\appendix
\fontsize{10}{10}
\selectfont
\section{Proofs}

We now provide proofs that were deferred from the main text. We start out with with the full proofs for the $\mov_C$ and $\mov_D$ objectives:

\noindent\textbf{Theorem \ref{theorem:mov-c}: }\textit{
	\textsc{MOVConstructive} obtains a $\frac{1}{3}\left(1 - \frac{1}{e}\right)$-approximation to the $\mov_C$ problem with any number of candidates. }
\begin{proof}
	
	Let $c(S, y) = \arg\min_{c_i} f(V_{c_*}^2 \cap V_{c_i}^1) - |V_{c_i}^1|$ be the candidate achieving the minimum in the definition of $m_C$. Let $S^*$ be an optimal seed set. Note that for all scenarios $y$, seed sets $S$, and candidates $c_i$, $f(S, y, V_{c_*}^2) \geq f(S, y, V_{c_*}^2 \cap V_{c_i}^1)$. Hence, we have
	
	%
	%
	%
	
	\begin{align*}
	\E_y\Big[f\left(S^*, y, V_{c_*}^2\right)\Big]  &\geq \frac{1}{3} \E_y\Big[f\left(S^*, y, V_{c_*}^1\right) + f\left(S^*, y, V_{c_*}^1 \cap V_{c(S^*, y)}^1\right) \\&+ f\left(S^*, y, V_{c_*}^1 \cap V_{c(S, y)}^1\right)\Big]\\
	\end{align*}
	
	Note that $\E_y[f(\cdot, y, V_{c_*}^2)]$ is a monotone submodular function, which \textsc{MOVConstructive} greedily maximizes. Let $S$ be the resulting seed set. We have 
	
	\begin{align*}
	\E_y\Big[f\left(S, y, V_{c_*}^2\right) + f\left(S, y, V_{c_*}^2 \cap V_{c(S, y)}^1\right)\Big] &\geq \E_y\Big[f\left(S, y, V_{c_*}^2\right)\Big]\\
	&\geq \frac{1}{3}\left(1 - \frac{1}{e}\right) \E_y\Big[f\left(S^*, y, V_{c_*}^2\right) +  f\left(S^*, y, V_{c_*}^2 \cap V_{c(S^*, y)}^1\right) \\&+ f\left(S^*, y, V_{c_*}^2 \cap V_{c(S, y)}^1\right)\Big]
	\end{align*}
	
	and so
	\allowdisplaybreaks
	\begin{align*}
	\mov_C(S) &= \E_y\Big[f\left(S, y, V_{c_*}^2\right) + \min_{c_j}f\left(S, y, V_{c_*}^2 \cap V_{c_j}^1\right) + \max_{c_i} |V_{c_i}^1| - |V_{c_j}^1|\Big]\\
	&=\E_y\Big[f\left(S, y, V_{c_*}^2\right) + f\left(S, y, V_{c_*}^2 \cap V_{c(S, y)}^1\right)\Big] + \max_{c_i} |V_{c_i}^1| - \E_y\Big[|V_{c(S, y)}^1|\Big]\\
	&\geq \frac{1}{3}\left(1 - \frac{1}{e}\right) \E_y\Big[f\left(S^*, y, V_{c_*}^1\right) + f\left(S^*, y, V_{c_*}^2 \cap V_{c(S^*, y)}^1\right) \\&\quad\quad\quad\quad+ f\left(S^*, y, V_{c_*}^2 \cap V_{c(S, y)}^1\right)\Big] + \max_{c_i} |V_{c_i}^1| - \E_y\Big[|V_{c(S, y)}^1|\Big]\\
	&\geq \frac{1}{3}\left(1 - \frac{1}{e}\right) \E_y\Big[f\left(S^*, y, V_{c_*}^1\right) + f\left(S^*, y, V_{c_*}^2 \cap V_{c(S^*, y)}^1\right) \\&\quad\quad\quad\quad+ f\left(S^*, y, V_{c_*}^1 \cap V_{c(S, y)}^1\right) + \max_{c_i} |V_{c_i}^1| - |V_{c(S, y)}^1|\Big]\\
	&=\frac{1}{3}\left(1 - \frac{1}{e}\right) \E_y\Big[f\left(S^*, y, V_{c_*}^1\right) + f\left(S^*, y, V_{c_*}^2 \cap V_{c(S^*, y)}^1\right)  \\&+f\left(S^*, y, V_{c_*}^1 \cap V_{c(S, y)}^1\right) + \max_{c_i} |V_{c_i}^1| - |V_{c(S, y)}^1| + |V_{c(S^*, y)}^1| - |V_{c(S^*, y)}^1|\Big]\\
	&=\frac{1}{3}\left(1 - \frac{1}{e}\right) \E_y\Big[f\left(S^*, y, V_{c_*}^1\right) + \min_{c_j} \left(f(S^*, y, V_{c_*}^2 \cap V_{c_j}^1) - |V_{c_j}^1|\right) \\&\quad\quad\quad\quad+ f(S^*, y, V_{c_*}^2 \cap V_{c(S, y)}^1)  + \max_{c_i} |V_{c_i}^1| - |V_{c(S, y)}^1| + |V_{c(S^*, y)}^1|\Big]\\
	&= \frac{1}{3}\left(1 - \frac{1}{e}\right)\Bigg(\mov_C(S^*) + \E_y\Big[f\left(S^*, y, V_{c_*}^2 \cap V_{c(S, y)}^1\right) + |V_{c(S^*, y)}^1| - |V_{c(S, y)}^1|\Big]\Bigg).
	\end{align*}
	
	Now by definition of $c(S^*, y)$, $f(S^*, y, V_{c_*}^2 \cap V_{c(S^*, y)}^1) - |V_{c(S^*, y)}^1| \leq f(S^*, y, V_{c_*}^2 \cap V_{c(S, y)}^1) - |V_{c(S, y)}^1|$ and so 
	
	\begin{align*}
	|V_{c(S^*, y)}^1|- |V_{c(S, y)}^1| \geq f\left(S^*, y, V_{c_*}^2 \cap V_{c(S^*, y)}^1\right) - f\left(S^*, y, V_{c_*}^2 \cap V_{c(S, y)}^1\right)
	\end{align*}
	
	This yields
	
	\begin{align*}
	\mov_C(S) &\geq \frac{1}{3}\left(1 - \frac{1}{e}\right)\Bigg(\mov_C(S^*) + \E_y\Big[f\left(S^*, y, V_{c_*}^2 \cap V_{c(S, y)}^1\right) + \\&\quad\quad\quad\quad f\left(S^*, y, V_{c_*}^2 \cap V_{c(S^*, y)}^1\right) - f\left(S^*, y, V_{c_*}^2 \cap V_{c(S, y)}^1\right)\Big]\Bigg)\\
	&=\frac{1}{3}\left(1 - \frac{1}{e}\right)\Bigg(\mov_C(S^*) + \E_y\Big[f\left(S^*, y, V_{c_*}^2 \cap V_{c(S^*, y)}^1\right)\Big]\Bigg)\\
	&\geq \frac{1}{3}\left(1 - \frac{1}{e}\right)\mov_C(S^*).
	\end{align*}

\end{proof}

\noindent\textbf{Theorem \ref{theorem:mov-d}: }\textit{
	\textsc{MOVDestructive} obtains a $\frac{1}{2}\left(1 - \frac{1}{e}\right)$-approximation to the multicandidate $\mov_D$ problem.}

\begin{proof}
	Now let $c(S, y) = \arg\max_{c_i} f(V_{c_*}^1 \cap V_{c_i}^2) + |V_{c_i}^1|$ be the candidate achieving the maximum in the definition of $m_D$. Let $S^*$ be an optimal seed set. Similarly to before, we have
	
	\begin{align*}
	\E_y\Big[f\left(S^*, y, V_{c_*}^1\right)\Big] \geq \frac{1}{2} \E_y\Big[f\left(S^*, y, V_{c_*}^1\right) + f\left(S^*, y, V_{c_*}^1 \cap V_{c(S^*, y)}^2\right)\Big].
	\end{align*}
	
	\textsc{MOVDestructive} greedily maximizes $\E_y\Big[f\left(S^*, y, V_{c_*}^1\right)\Big]$. Call the resulting seed set $S$. We have
	\allowdisplaybreaks
	\begin{align*}
	\mov_D(S) &= \E_y\Big[f\left(S, y, V_{c_*}^1\right) + f\left(S, y, V_{c_*}^1 \cap V_{c(S, y)}^2\right) + |V_{c(S, y)}^1| - \max_{c_i}|V_{c_i}^1|\Big]\\
	&\geq \left(1 - \frac{1}{e}\right)\E_y\Big[f\left(S^*, y, V_{c_*}^1\right)\Big] + \E_y\Big[f\left(S, y, V_{c_*}^1 \cap V_{c(S, y)}^2\right) + |V_{c(S, y)}^1| - \max_{c_i}|V_{c_i}^1|\Big]\\
	&\geq \frac{1}{2}\left(1 - \frac{1}{e}\right)\E_y\Big[f\left(S^*, y, V_{c_*}^1\right)  + f\left(S^*, y, V_{c_*}^1 \cap V_{c(S^*, y)}^2\right)\Big] \\&\quad\quad\quad\quad\quad+\E_y\Big[f\left(S, y, V_{c_*}^1 \cap V_{c(S, y)}^2\right) + |V_{c(S, y)}^1| - \max_{c_i}|V_{c_i}^1|\Big]\\
	&\geq \frac{1}{2}\left(1 - \frac{1}{e}\right)\E_y\Big[f\left(S^*, y, V_{c_*}^1\right)  + f\left(S^*, y, V_{c_*}^1 \cap V_{c(S^*, y)}^2\right) \\&\quad\quad\quad\quad\quad+f\left(S, y, V_{c_*}^1 \cap V_{c(S, y)}^2\right) + |V_{c(S, y)}^1| - \max_{c_i}|V_{c_i}^1|\Big]\\
	&\geq \frac{1}{2}\left(1 - \frac{1}{e}\right)\E_y\Big[f\left(S^*, y, V_{c_*}^1\right)  + f\left(S^*, y, V_{c_*}^1 \cap V_{c(S^*, y)}^2\right) \\&\quad\quad\quad\quad\quad+f\left(S, y, V_{c_*}^1 \cap V_{c(S, y)}^2\right) + |V_{c(S, y)}^1| + |V_{c(S^*, y)}^1| - |V_{c(S^*, y)}^1|- \max_{c_i}|V_{c_i}^1|\Big]\\
	&\geq \frac{1}{2}\left(1 - \frac{1}{e}\right)\Big[\mov_D(S^*) +\E_y \Big[f\left(S, y, V_{c_*}^1 \cap V_{c(S, y)}^2\right) + |V_{c(S, y)}^1| - |V_{c(S^*, y)}^1|\Big]\Big].
	\end{align*}
	
	Now using the definition of $c(S, y)$, we have that $f(S, y, V_{c_*}^1 \cap V_{c(S, y)}^2) + |V_{c(S, y)}^1| \geq f(S, y, V_{c_*}^1 \cap V_{c(S^*, y)}^2) + |V_{c(S^*, y)}^1|$. This yields
	
	\begin{align*}
	|V_{c(S, y)}^1| - |V_{c(S^*, y)}^1| \geq f\left(S, y, V_{c_*}^1 \cap V_{c(S^*, y)}^2\right) - f\left(S, y, V_{c_*}^1 \cap V_{c(S, y)}^2\right)
	\end{align*}
	
	and so we have
	
	\begin{align*}
	\mov_D(S) &\geq \frac{1}{2}\left(1 - \frac{1}{e}\right)\Big[\mov_D(S^*) +\E_y \Big[f\left(S, y, V_{c_*}^1 \cap V_{c(S^*, y)}^2\right)\Big]\Big]\\
	&\geq \frac{1}{2}\left(1 - \frac{1}{e}\right)\mov_D(S^*)
	\end{align*}
	
\end{proof}

We now prove corresponding bicriteria guarantees for the $\pov$ objectives.

\noindent\textbf{Theorem \ref{theorem:povc}: }\textit{
	Let $OPT(\Delta)$ denote the optimal value of the problem $\max_{|S| \leq k} \pr_y\left[m_C(S, y) \geq \Delta\right]$. Let $S$ be the set produced by \textsc{POVConstructive}. We have
	\begin{align*}
	\pov_C(S) \geq \max_{0 < \alpha < 1} \frac{\frac{e-1}{3e-1} OPT\left(\frac{1}{\alpha} \Delta_C\right) - \alpha}{1 - \alpha}
	\end{align*} }

\begin{proof} 
	
	The main difference from the two candidate case is that $\frac{1}{m}\sum_y \min\left(\beta, m_C(S, y)\right)$ is no longer submodular since $m_C$ need not be a submodular function. However, the proof of Theorem 4.3 only uses submodularity in establishing an approximation guarantee for greedy optimization of the surrogate. In the multicandidate case, we will greedily optimize $\frac{1}{m}\sum_y \min\left(\beta, f(S, y, V_{c_*}^2)\right)$, which is submodular. Let $S_\beta$ be the resulting seed set and $S^*$ be a set that optimizes $\frac{1}{m}\sum_y \min\left(\beta, m_C(S, y)\right)$. If we can prove that $\frac{1}{m}\sum_y \min\left(\beta, m_C(S_\beta, y)\right) \geq  \gamma\frac{1}{m}\sum_y \min\left(\beta, m_C(S^*, y)\right)$ for some constant factor $\gamma$, then the same argument as in Theorem 4.3 extends to the multicandidate case. Fix any particular value of $\beta$. We establish a constant factor approximation as follows:
	
	\begin{align*}
	\frac{1}{m}\sum_y \min\left(\beta, m_C(S^*, y)\right) &=\frac{1}{m}\sum_y \min\left(\beta, m_C(S_\beta, y) + m_C(S^*, y) - m_C(S_\beta, y)\right)\\
	&\leq \frac{1}{m}\sum_y \min\Big(\beta, m_C(S_\beta, y) + f\left(S^*, y, V_{c_*}^2\right) - f\left(S_\beta, y, V_{c_*}^2\right) \\&\quad\quad\quad\quad\quad+ f\left(S^*, y, V_{c_*}^2 \cap V_{c(S^*,y)}^1\right) - f\left(S_\beta, y, V_{c_*}^2 \cap V_{c(S_\beta,y)}^1\right) + |V_{c(S_\beta, y)}^1| -  |V_{C(S^*, y)}^1|\Big)
	\end{align*}
	
	Via the definition of $c(S^*, y)$, we have that $|V_{c(S_\beta, y)}^1| -  |V_{c(S^*, y)}^1| \leq f\left(S^*, y, V_{c_*}^2 \cap V_{c(S_\beta, y)}^1\right) - f\left(S^*, y, V_{c_*}^2 \cap V_{c(S^*, y)}^1\right)$. This yields
	\allowdisplaybreaks
	\begin{align*}
	\frac{1}{m}\sum_y \min\left(\beta, m_C(S^*, y)\right) &\leq \frac{1}{m}\sum_y \min\Big(\beta, m_C(S_\beta, y) + f\left(S^*, y, V_{c_*}^2\right) + f\left(S^*, y, V_{c_*}^2 \cap V_{c(S_\beta, y)}^1\right)\Big)\\
	&\leq \frac{1}{m}\sum_y \min\Big(\beta, m_C(S_\beta, y) + 2f\left(S^*, y, V_{c_*}^2\right) \Big)\\
	&\leq \frac{1}{m}\sum_y \min\Big(\beta, m_C(S_\beta, y) \Big) + \frac{2}{m}\sum_y \min\Big(\beta, f\left(S^*, y, V_{c_*}^2\right) \Big)\\
	&\leq \frac{1}{m}\sum_y \min\Big(\beta, m_C(S_\beta, y) \Big) + \frac{1}{m}\frac{2e}{e-1}\sum_y \min\Big(\beta, f\left(S_\beta, y, V_{c_*}^2\right) \Big)\\
	&\leq \frac{1}{m}\sum_y \min\Big(\beta, m_C(S_\beta, y) \Big) + \frac{1}{m}\frac{2e}{e-1}\sum_y \min\Big(\beta, m_C(S_\beta, y) \Big)\\
	&\leq \left(1 + \frac{2e}{e-1}\right) \frac{1}{m}\sum_y \min\Big(\beta, m_C(S_\beta, y) \Big)
	\end{align*}
	
	and now the conclusion follows by applying the same argument as in Theorem \ref{theorem:povbicriteria}. 
\end{proof}

\noindent\textbf{Theorem \ref{theorem:povd}: }\textit{
	Let $OPT(\Delta)$ denote the optimal value of the problem $\max_{|S| \leq k} \pr_y\left[m_D(S, y) \geq \Delta\right]$. Let $S$ be the set produced by \textsc{POVDestructive}. We have
	\begin{align*}
	\pov_D(S) \geq \max_{0 < \alpha < 1} \frac{\frac{e-1}{3e-1} OPT(\frac{1}{\alpha} \Delta) - \alpha}{1 - \alpha}
	\end{align*} }

\begin{proof}
	Applying the reasoning as in Theorem \ref{theorem:povc}, we have
	\begin{align*}
	\frac{1}{m}\sum_y \min\left(\Delta, m_D(S^*, y)\right) &=\frac{1}{m}\sum_y \min\left(\Delta, m_C(S, y) + m_D(S^*, y) - m_D(S, y)\right)\\
	&\leq \frac{1}{m}\sum_y \min\Big(\Delta, m_D(S, y) + f\left(S^*, y, V_{c_*}^1\right) - f\left(S, y, V_{c_*}^1\right) \\&+ f\left(S^*, y, V_{c_*}^1 \cap V_{c(S^*,y)}^2\right) - f\left(S, y, V_{c_*}^1 \cap V_{c(S,y)}^2\right) + |V_{C(S^*, y)}^1| -  |V_{C(S, y)}^1| \Big).
	\end{align*}
	
	The definition of $c(S, y)$ implies that $$|V_{c(S^*, y)}^1| - |V_{c(S, y)}^1 \leq f\left(S, y, V_{c_*}^1 \cap V_{c(S, y)}^2\right) -  f\left(S, y, V_{c_*}^1 \cap V_{c(S^*, y)}^2\right)$$ so we have
	\\
	\begin{align*}
	\frac{1}{m}\sum_y \min\left(\Delta, m_D(S^*, y)\right)
	&\leq \frac{1}{m}\sum_y \min\Big(\Delta, m_D(S, y) + f\left(S^*, y, V_{c_*}^1\right) + f\left(S^*, y, V_{c_*}^1 \cap V_{c(S^*,y)}^2\right) \Big)\\
	&\leq \frac{1}{m}\sum_y \min\Big(\Delta, m_D(S, y) + 2f\left(S^*, y, V_{c_*}^1\right)\Big)
	\end{align*}
	
	and now the theorem follows from the same argument as in Theorem \ref{theorem:povc}.
\end{proof}

\section{Additional experimental results}

\begin{figure}[h]
	\centering
	\includegraphics[width=2in]{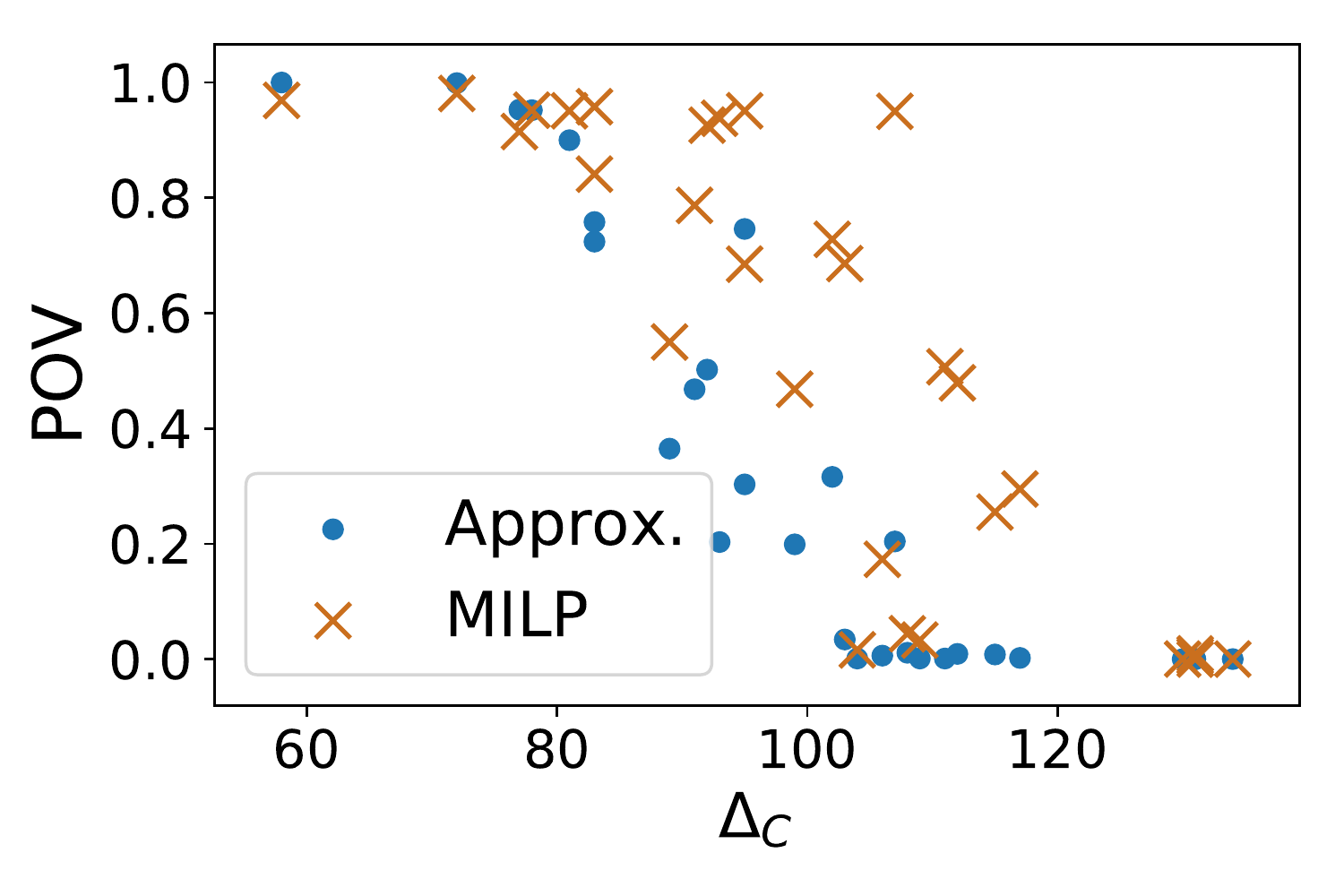}
	\includegraphics[width=2in]{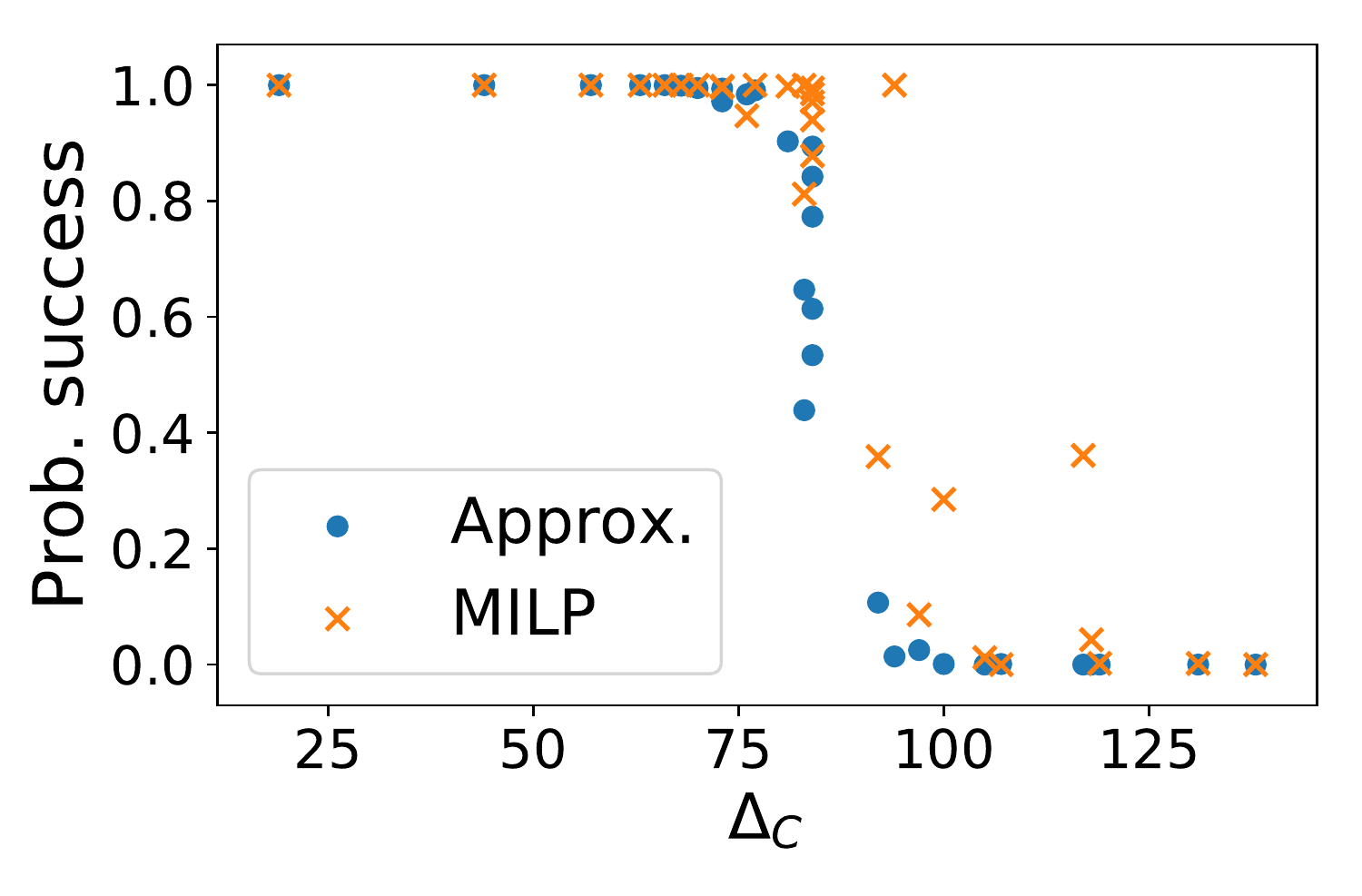}
	\caption{Probability of victory in constructive control. Left: irvine. Right: facebook}
\end{figure}

\begin{figure}[h]
	\centering
	\includegraphics[width=2in]{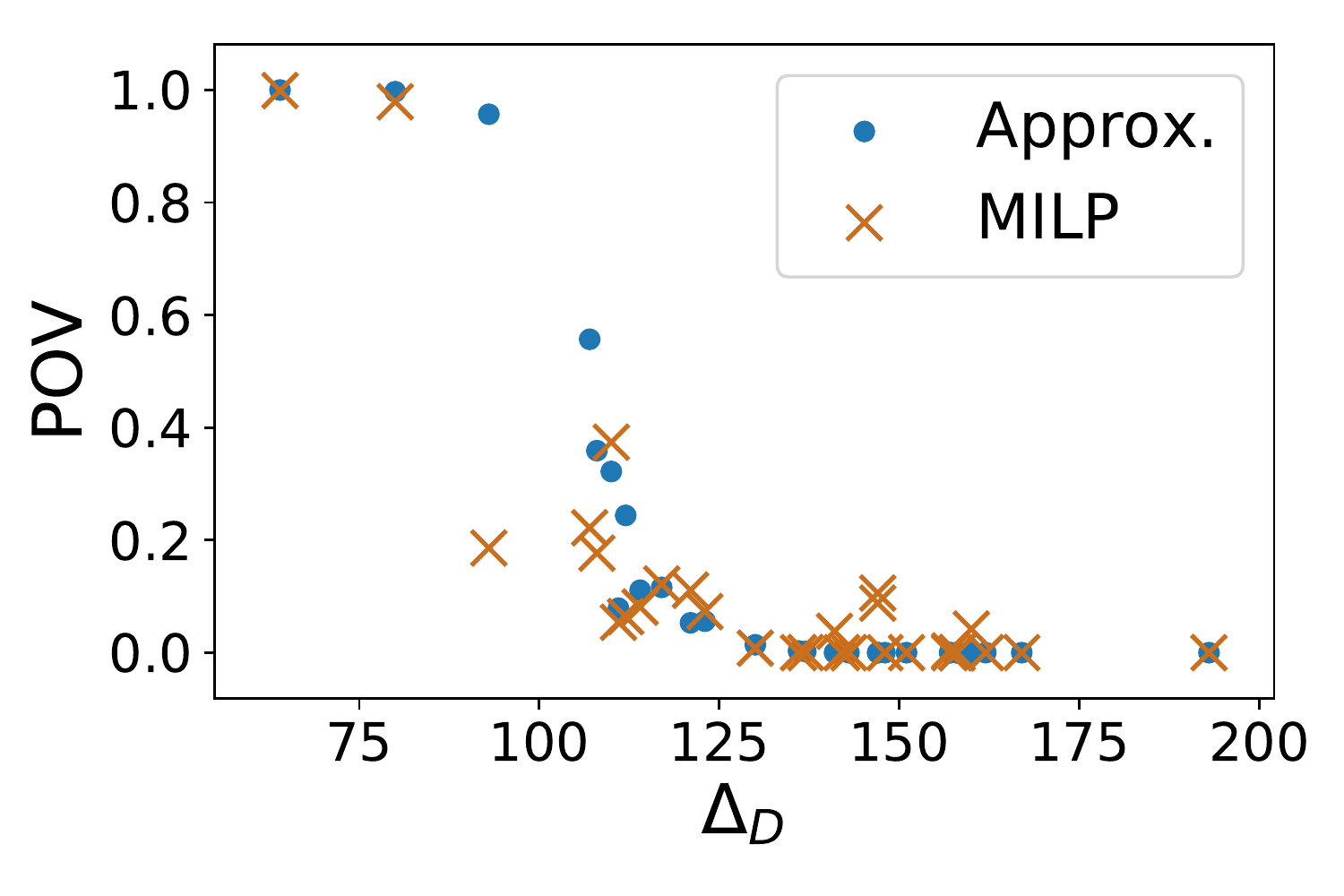}
	\includegraphics[width=2in]{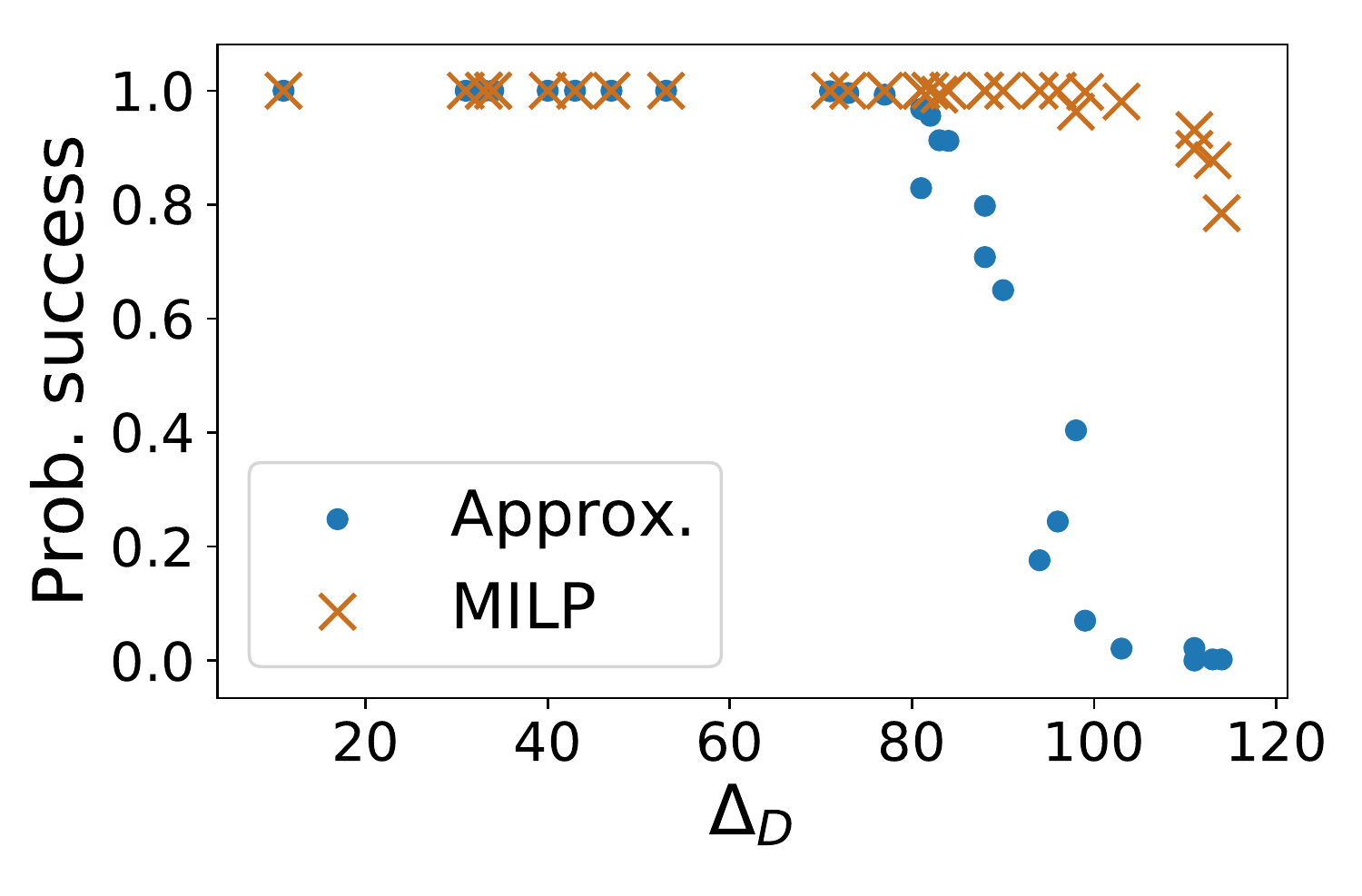}
	\caption{Probability of victory in destructive control. Left: irvine. Right: facebook. On irvine, the MILP was terminated after 24 hours, and had not found competitive solutions with the approximation algorithm on the intermediate margin instances by that time.}
\end{figure}

\end{document}